\documentclass[11pt,a4paper]{article}
\usepackage[margin=25mm]{geometry}
\usepackage[T1]{fontenc}
\usepackage{lmodern}
\usepackage{amsmath,amssymb,amsthm,mathtools}
\usepackage{microtype,needspace,cite,etoolbox}
\apptocmd{\thebibliography}{\setlength{\itemsep}{0pt}}{}{}
\usepackage[dvipsnames]{xcolor}
\usepackage[colorlinks=true,linkcolor=MidnightBlue,citecolor=MidnightBlue,urlcolor=MidnightBlue]{hyperref}
\usepackage[capitalise,nameinlink,noabbrev]{cleveref}
\usepackage{booktabs,array,pgfplots,placeins}
\usepackage{authblk}
\pgfplotsset{compat=1.18}
\newtheorem{theorem}{Theorem}[section]
\newtheorem{lemma}[theorem]{Lemma}
\newtheorem{proposition}[theorem]{Proposition}
\newtheorem{corollary}[theorem]{Corollary}
\theoremstyle{remark}

\numberwithin{equation}{section}
\newcommand{\cA}{\mathcal A}
\newcommand{\cN}{\mathcal N}
\newcommand{\cM}{\mathcal M}
\newcommand{\cL}{\mathcal L}
\newcommand{\cD}{\mathcal D}
\newcommand{\cS}{\mathcal S}
\newcommand{\cF}{\mathcal F}
\newcommand{\Hraw}{H_0}
\newcommand{\Tr}{\operatorname{Tr}}
\newcommand{\rank}{\operatorname{rank}}
\newcommand{\supp}{\operatorname{supp}}
\newcommand{\diag}{\operatorname{diag}}
\newcommand{\id}{\operatorname{id}}
\newcommand{\ran}{\operatorname{range}}
\newcommand{\ket}[1]{\lvert #1\rangle}
\newcommand{\bra}[1]{\langle #1\rvert}
\newcommand{\proj}[1]{\lvert #1\rangle\!\langle #1\rvert}
\newcommand{\normF}[1]{\lVert #1\rVert_F}
\newcommand{\normop}[1]{\lVert #1\rVert_{\mathrm{op}}}
\newcommand{\EF}{E_{F}}
\newcommand{\EC}{E_{C}}
\newcommand{\Ech}{E_{F}^{\mathrm{ch}}}
\newcommand{\Epar}{E_{C}^{\mathrm{par}}}

\newcommand{\roof}[1]{\widehat S_{#1}}

\title{Classical Capacity and Entanglement Cost of the Amplitude Damping Channel}

\author[1]{Ziao Tang}
\author[2]{Chengkai Zhu}
\author[1,*]{Ge Bai}
\author[1,$\dagger$]{Xin Wang}

\affil[1]{The Hong Kong University of Science and Technology (Guangzhou),
Guangdong 511453, China}

\affil[2]{QudeLeap Research, Shanghai 200030, China}

\begin{document}
\maketitle
\begingroup
\renewcommand{\thefootnote}{\fnsymbol{footnote}}
\footnotetext[1]{\href{mailto:gebai@hkust-gz.edu.cn}
{gebai@hkust-gz.edu.cn}}

\footnotetext[2]{\href{mailto:felixxinwang@hkust-gz.edu.cn}
{felixxinwang@hkust-gz.edu.cn}}
\endgroup
\begin{abstract}
Determining a noisy quantum channel's classical capacity and entanglement cost generally requires regularization
over many channel uses. We remove both regularizations for every
qubit-to-qubit channel admitting a pure output. For each such
channel, Holevo information, channel entanglement of formation,
and parallel entanglement cost are additive with those of any
finite-dimensional partner channel. This class includes all
qubit-to-qubit channels of Kraus rank at most two. For the amplitude damping channel with damping
probability $p$, the unassisted classical capacity equals the known
single-use Holevo information, attained by a binary pure-state ensemble with collective decoding, and the entanglement cost is
$h_2((1+\sqrt p)/2)$ ebits per use. The common mechanism is a support criterion for strong superadditivity of entanglement of
formation: one marginal has no support on the sector in which both local systems are orthogonal to fixed distinguished vectors. We
prove this criterion in arbitrary finite dimensions using a triangular block-matrix entropy inequality and decompositions preserving two expectations. For amplitude damping, we also derive an exact finite-block Holevo deficit, identify the unique optimal average input for $p<1$, and construct a binary Kraus representation attaining the uniform formation bound.
\end{abstract}

\setcounter{tocdepth}{1}
\tableofcontents

\section{Introduction}\label{sec:intro}
A noisy quantum channel can be studied through the information
it transmits and the resources required to simulate it. These
questions have different operational meanings, but share a
mathematical obstruction: correlations across channel uses can
change the asymptotic rate. For a finite-dimensional memoryless
channel $\cN$, the Holevo--Schumacher--Westmoreland theorem gives
the unassisted classical capacity as
\begin{equation}\label{eq:hsw}
 C(\cN)=\lim_{n\to\infty}\frac1n\chi(\cN^{\otimes n}),
\end{equation}
where $\chi(\cN)$ is the optimized Holevo information
\cite{holevo,schumacher-westmoreland}. Product-state encodings
achieve $\chi(\cN)$; identifying this rate with the capacity
requires a converse allowing entangled codewords. The parallel
entanglement cost of simulation by local operations and classical
communication (LOCC) is likewise a regularized quantity, with
channel entanglement of formation replacing Holevo information
\cite{op:BBCW}. Additivity removes these regularizations. Since
it fails for general channels \cite{hastings}, a single-use
characterization needs a structural reason specific to the channel.

The qubit amplitude-damping channel models relaxation from an
excited state to a ground state and provides a basic nonunital
test case. Its quantum and entanglement-assisted classical
capacities have single-letter expressions
\cite{giovannetti-fazio,khatri,op:DS,op:BSST}. Giovannetti and
Fazio evaluated the single-use Holevo information and its
optimal binary pure-state ensemble \cite{giovannetti-fazio};
the same binary-ensemble optimization appears in the earlier
analysis of Bennett et al.~\cite[Sec.~III.B]{op:BSST}. The distinction between this achievable
rate and the unassisted capacity is the possible advantage of
entangled encodings. Additivity for entanglement-breaking
channels \cite{shor-eb} and unital qubit channels \cite{king}
does not settle the nontrivial amplitude-damping regime
$0<p<1$.

Several approaches bound this possible advantage. Brand\~ao
et al. subtract an output--environment correlation penalty from
an output-entropy bound \cite{brandao}. Semidefinite-programming
converses \cite{wang-xie-duan}, approximate-additivity estimates
\cite{leditzky}, and symmetry-reduced hierarchies \cite{fawzi}
provide other ways to control the regularized capacity. Khatri,
Sharma, and Wilde compare capacity bounds for generalized amplitude damping \cite{khatri}. More recently, Fang makes a correlation-based converse computable for that family and obtains close numerical agreement with the Holevo information at positive temperature~\cite{fang2026}. However, it is still unclear whether its known single-use rate is exact even with arbitrary entangled encodings.

For simulation, a related gap occurs between two state quantities.
If $J_{\cN}$ is the normalized Choi state, every qubit-to-qubit
channel satisfies
\begin{equation}\label{eq:intro-sandwich}
 \EC(J_{\cN})\le\Epar(\cN)\le\EF(J_{\cN}),
\end{equation}
where $\EF$ and $\EC$ are state entanglement of formation and
state entanglement cost, and $\Epar$ is parallel channel
entanglement cost \cite{op:BBCW,op:Konrad}. The upper bound
is computable from a single two-qubit state, but the lower bound
is regularized. General lower bounds on channel entanglement
cost provide complementary estimates
\cite{lami-regula2023,wang-jing-zhu2025}. To make
\eqref{eq:intro-sandwich} exact, it suffices to establish
$\EC(J_{\cN})=\EF(J_{\cN})$. Thus both operational questions
lead to additivity, evaluated on different state families.

In this work, we remove both regularizations for every
qubit-to-qubit channel admitting a pure output. The common
mechanism is a support constraint: two local systems cannot both
be orthogonal to fixed distinguished vectors. This constraint
implies strong superadditivity of entanglement of formation,
including for correlated extensions with arbitrary
finite-dimensional auxiliary systems. It can be applied either
to output--environment states of a Stinespring dilation or to
the reference--output Choi state.

First, we prove the support criterion in arbitrary finite
dimensions (\cref{thm:support}). For a pure extension, the
excluded joint sector produces a triangular coefficient matrix.
A matrix entropy inequality (\cref{thm:triangle}) separates
its entanglement into a scalar contribution and weighted block
entropies. Decompositions preserving two expectations
\cite{leka-petz,leka} bound the first marginal's formation by
the scalar term; local monotonicity and convexity bound the
second marginal's formation by the block terms. This proves
the required inequality without evaluating either mixed-state
convex roof. The constraint concerns support and allows
coherence between the permitted sectors.

Second, we derive the two channel consequences. The
Matsumoto--Shimono--Winter correspondence \cite{msw} converts
the support criterion for a Stinespring image into strong
Holevo additivity (\cref{cor:stinespring}). For qubit-to-qubit
channels, this support condition is equivalent to admitting
a pure output (\cref{lem:pure-output-stinespring}) and to the
existence of a product vector in the kernel of the Choi state
(\cref{sec:simulation-criterion}). The latter condition makes
the bounds in \eqref{eq:intro-sandwich} coincide. Consequently,
Holevo information, channel entanglement of formation, and
parallel entanglement cost are additive with arbitrary partner
channels throughout the pure-output class
(\cref{thm:two-kraus-unified,cor:pure-output}). This class
includes every qubit-to-qubit channel of Kraus rank at most
two (\cref{cor:two-kraus}). These applications use the relation
between channel and state additivity developed in
\cite{msw,shor-equivalence}; their support condition is
different from an entanglement-breaking-subspace condition
\cite{zhao-chen}.

For amplitude damping, the results identify the known
single-use Holevo optimization with the unassisted classical
capacity and give the parallel simulation cost
$h_2((1+\sqrt p)/2)$ ebits per use
(\cref{thm:damping-capacity,thm:ad-cost}). An exact finite-block
Holevo deficit identifies the unique optimal average input
for $p<1$ (\cref{prop:deficit}). A binary Kraus representation
also attains a formation upper bound uniformly over all
block inputs (\cref{prop:binary-instrument}). The new
operational conclusions come from the matching many-use
converses, rather than from a new single-use capacity formula
or a new two-qubit formation formula.

The support criterion explains why these regularizations
collapse and gives a concrete route to look for further
additive channels: identify the excluded sector in the
appropriate bipartition. Extending the argument to finite
temperature would require controlling a nonzero weight in
that sector. Throughout this work, simulation approximates
tensor-power channels in diamond norm, with classical
communication free. Simulation against adaptive tests is
a distinct task \cite{op:Wilde2018,lami-regula2023}.

\paragraph{Organization.}
\Cref{sec:setup,sec:results} specify the operational models
and summarize the channel conclusions. \Cref{sec:support-proof}
proves the support theorem; \cref{sec:capacity,sec:simulation}
derive its communication and simulation consequences.
The appendices give the matrix proof, concurrence calculations,
uniform Kraus argument, and comparisons with established rates.

\section{Preliminaries}\label{sec:setup}
\subsection{Notation}
All Hilbert spaces are finite dimensional. We write $\cD(X)$
and $\cL(X)$ for the density operators and linear operators on
$X$, respectively. A quantum channel is a completely positive,
trace-preserving linear map. Its Kraus rank is the minimum number
of operators in a Kraus representation. A channel admits a pure
output if $\cN(\proj\psi)=\proj\phi$ for some unit vectors
$\psi,\phi$. When a pure vector appears as a
state or a channel argument, it denotes its rank-one projector.
Entanglement across a bipartition $A:B$ is indicated by a colon.

The von Neumann entropy is $S(\rho)=-\Tr\rho\log_2\rho$,
and the binary entropy is defined by
$h_2(x)=-x\log_2x-(1-x)\log_2(1-x)$, with $0\log_2 0=0$.
We use base-two logarithms unless otherwise specified; $\ln$
denotes the natural logarithm. The quantum relative entropy is
\begin{equation}
 D(\rho\Vert\sigma)=\Tr\rho(\log_2\rho-\log_2\sigma),
\end{equation}
with value $+\infty$ when $\supp\rho\nsubseteq\supp\sigma$.
For a matrix $M$, the symbols $M^\dagger$ and $M^*$ denote the adjoint and entrywise conjugate. Its Frobenius and operator
norms are denoted by $\normF M=(\Tr M^\dagger M)^{1/2}$
and $\normop M$. For later use, define the scalar function
\begin{equation}\label{eq:E-function}
 e(c)=h_2\!\left(\frac{1+\sqrt{1-c^2}}2\right),
 \qquad 0\le c\le1.
\end{equation}
For a two-qubit state, this is the function relating concurrence
to entanglement of formation in Wootters' formula~\cite{wootters}. It is increasing and convex on $[0,1]$, as recalled in
\cref{app:scalar}.

The qubit amplitude-damping channel with damping probability
$p\in[0,1]$ is defined by
\begin{equation}\label{eq:ad-kraus}
 \cA_p(\tau)=K_0\tau K_0^\dagger+K_1\tau K_1^\dagger,
 \qquad
 K_0=\begin{pmatrix}1&0\\0&\sqrt{1-p}\end{pmatrix},\quad
 K_1=\begin{pmatrix}0&\sqrt p\\0&0\end{pmatrix}.
\end{equation}
An Stinespring isometry $V_p:X\to B\otimes E$, with qubit
input $X$, output $B$, and environment $E$, is
\begin{equation}\label{eq:dilation}
 V_p\ket0=\ket{0_B0_E},\qquad
 V_p\ket1=\sqrt{1-p}\ket{1_B0_E}+\sqrt p\ket{0_B1_E}.
\end{equation}
Its image is orthogonal to $\ket{1_B1_E}$, a support constraint
that will imply strong Holevo additivity.

\subsection{Unassisted classical communication}
A length-$n$ code for $\cN:X\to B$ consists of input states
$\rho_m\in\cD(X^{\otimes n})$ and a decoding POVM
$\{\Lambda_m\}_{m=1}^{M_n}$ on $B^{\otimes n}$. For
equiprobable messages, its average error is
\[
 \epsilon_n=1-\frac1{M_n}\sum_{m=1}^{M_n}
       \Tr[\Lambda_m\cN^{\otimes n}(\rho_m)].
\]
The unassisted classical capacity is the supremum of rates
achievable by code sequences for which $\epsilon_n\to0$, with
rate measured as $n^{-1}\log_2M_n$ bits per use. Encodings may
be entangled across channel uses, and decoding may be collective.
No preshared entanglement or auxiliary classical communication
is available.

At a fixed average input $\tau$, define
\begin{equation}\label{eq:output-roof}
 \roof{\cN}(\tau)=
 \min_{\tau=\sum_jw_j\proj{\psi_j}}
       \sum_jw_jS(\cN(\psi_j)),\qquad
 \chi_{\cN}(\tau)=S(\cN(\tau))-\roof{\cN}(\tau).
\end{equation}
The minimum is over finite ensembles of normalized pure states,
with $w_j\ge0$ and $\sum_jw_j=1$. Thus $\roof{\cN}$ is
the convex roof of the output entropy, and $\chi_{\cN}(\tau)$
is the largest Holevo information of an ensemble averaging to
$\tau$. The former is convex and the latter concave in $\tau$.
Optimizing over the average input gives
$\chi(\cN)=\max_\tau\chi_{\cN}(\tau)$.

\subsection{Parallel channel simulation}
Alice holds the input, and Bob is required to receive the output.
They have access to LOCC and a shared maximally entangled state
$\Phi_{M_n}$, where $\Phi_M$ is the projector onto
$M^{-1/2}\sum_{j=1}^M\ket j\ket j$. An LOCC simulation
protocol $\Lambda_n$ induces a channel satisfying
\begin{equation}\label{eq:simulation-model}
 \widetilde\cN_n(\xi)=\Lambda_n(\xi\otimes\Phi_{M_n}),
 \qquad
 \frac12\|\widetilde\cN_n-\cN^{\otimes n}\|_\diamond\longrightarrow0.
\end{equation}
The parallel entanglement cost $\Epar(\cN)$ is the infimum,
over such protocol sequences, of
$\limsup_{n\to\infty}n^{-1}\log_2M_n$
\cite{op:BBCW,op:Wilde2018}. It is measured in ebits per use,
with classical communication free. Diamond-norm convergence
requires the simulator to approximate the target uniformly over
all inputs, including those correlated across uses and with an
arbitrary reference. The target is the tensor-power channel;
the definition does not include adaptive interleaving of channel
calls.

For a bipartite state $\rho_{AB}$, the entanglement cost $\EC(\rho)$ is the minimum asymptotic number of shared ebits per copy needed to prepare $\rho_{AB}^{\otimes n}$ by LOCC
with vanishing trace-distance error. Define its entanglement of formation by
\begin{equation}\label{eq:state-roof}
 \EF(\rho_{AB})=
 \min_{\rho=\sum_jw_j\proj{\psi_j}}
       \sum_jw_j S(\Tr_B\proj{\psi_j}).
\end{equation}
Its regularization characterizes the state entanglement cost~\cite{hht}:
\begin{equation}\label{eq:state-cost}
\EC(\rho)=\lim_{n\to\infty}\frac1n\EF(\rho^{\otimes n}).
\end{equation}
For $d_X=\dim X$, the normalized Choi state of $\cN$ is
$J_{\cN}=(\id\otimes\cN)(\Phi_{d_X})$. The channel
entanglement of formation is~\cite{op:BBCW}
\begin{equation}\label{eq:channel-formation}
 \Ech(\cN)=\max_{\psi_{RX}\ \mathrm{pure}}
       \EF\!\left((\id_R\otimes\cN)(\psi_{RX})\right),
 \qquad R\cong X.
\end{equation}
Here, the entanglement is evaluated across $R:B$. The
regularization of $\Ech$ characterizes $\Epar$; we state
the precise relation in \eqref{eq:berta-regularization}.
The regularized limits used here exist: Holevo information
is superadditive under tensor products, while state entanglement
of formation is subadditive. Channel entanglement of formation
is also subadditive, as shown in \cref{lem:channel-subadditivity}.
These quantities have linear dimension bounds, so Fekete's
lemma applies to their tensor-power sequences.

\section{Main results}\label{sec:results}
We summarize the channel conclusions before proving them in
\cref{sec:capacity,sec:simulation}. The first result identifies a common class for classical communication and channel simulation. The amplitude-damping channel then admits explicit
evaluations of both operational rates. All partner channels
below have arbitrary finite input and output dimensions.

\paragraph{Main result 1: qubit channels admitting a pure output
(see \cref{cor:pure-output}).}
For every qubit-to-qubit channel $\cN$ admitting a pure output,
\[
 C(\cN)=\chi(\cN),\qquad
 \Epar(\cN)=\EF(J_{\cN}).
\]
Moreover, for every finite-dimensional channel $\cM$,
\begin{align*}
 \chi(\cN\otimes\cM)&=\chi(\cN)+\chi(\cM),\\
 C(\cN\otimes\cM)&=C(\cN)+C(\cM),\\
 \Ech(\cN\otimes\cM)&=\Ech(\cN)+\Ech(\cM),\\
 \Epar(\cN\otimes\cM)&=\Epar(\cN)+\Epar(\cM).
\end{align*}
The result also identifies the Choi-state cost and the channel
formation quantity with the same value and gives the corresponding
tensor-power identities. It removes the many-use regularizations;
the single-use ensemble optimization defining $\chi(\cN)$ remains
for a general channel in this class. The communication statement
follows from \cref{lem:pure-output-stinespring,cor:stinespring}.
\Cref{thm:two-kraus-unified} collects the communication,
formation, and simulation conclusions for Kraus rank at most
two; \cref{cor:pure-output} extends all of them to the
pure-output class using the equivalent support and Choi-state
criteria.

\paragraph{Main result 2: classical capacity of amplitude damping
(see \cref{thm:damping-capacity}).}
For every $p\in[0,1]$,
\[
 C(\cA_p)=\max_{0\le q\le1}
       \left[h_2((1-p)q)-h_2\left(\frac{1}{2}+\frac{\sqrt{1-4p(1-p)q^2}}{2}\right)\right].
\]
For an optimizing $q_*$, the equiprobable states
$\sqrt{1-q_*}\ket0\pm\sqrt{q_*}\ket1$ yield optimal
product encodings with collective decoding. This scalar optimization and binary alphabet appear in
\cite[Sec.~III.B]{op:BSST} and
\cite[Eqs.~(42)--(43)]{giovannetti-fazio}. The new converse identifies
their rate with the unassisted capacity even when codewords
are entangled across uses. The finite-block refinement in
\cref{prop:deficit} determines the unique optimal average input
at every blocklength for $p<1$.

\paragraph{Main result 3: entanglement cost of amplitude damping
(see \cref{thm:ad-cost}).}
Under the parallel LOCC simulation model of
\eqref{eq:simulation-model},
\[
 \Epar(\cA_p)=h_2\!\left(\frac{1+\sqrt p}{2}\right).
\]
This is also the entanglement of formation and entanglement
cost of $J_{\cA_p}$. The single-copy formation value and the
channel-simulation upper bound are established results
\cite{wootters,op:BBCW}; strong Choi-state additivity supplies
the matching lower bound. \Cref{thm:ad-cost} also gives exact
finite-copy channel formation and arbitrary-partner additivity.
\Cref{prop:binary-instrument} provides an explicit two-outcome
Kraus representation attaining the required uniform bound.

\Cref{fig:ad-rates} compares the two rates determined here
with the established quantum capacity $Q$ and
entanglement-assisted classical capacity $C_E$
\cite{giovannetti-fazio,khatri,op:BSST,op:DS}. 

\begin{figure}[tbp]
\centering
\pgfplotstableread[row sep=\\]{
p C Q cost CE\\
0 1 1 1 2\\
0.005 0.986268799202 0.974843103377 0.996390250756 1.97480637424\\
0.01 0.975165289823 0.954716150941 0.992774453988 1.95461339555\\
0.015 0.965080752987 0.936474477566 0.989152573042 1.93629211381\\
0.02 0.955662546339 0.919444872114 0.985524570864 1.91917500928\\
0.025 0.946740512156 0.903306846731 0.981890409996 1.90294468961\\
0.03 0.938213251762 0.887871445905 0.978250052568 1.88741413916\\
0.035 0.930013151273 0.873013787887 0.974603460292 1.87245979142\\
0.04 0.922091803669 0.858645144252 0.970950594455 1.85799385611\\
0.045 0.914412801042 0.844699221601 0.967291415912 1.84395073313\\
0.05 0.906947747394 0.83112461607 0.96362588508 1.83027954444\\
0.055 0.899673868627 0.817880314113 0.959953961929 1.81693968321\\
0.06 0.892572490493 0.804932841164 0.956275605978 1.80389799424\\
0.065 0.885628022216 0.792254366236 0.952590776282 1.79112690027\\
0.07 0.878827251343 0.779821392497 0.948899431429 1.77860310789\\
0.075 0.872158838865 0.767613823488 0.945201529531 1.76630668482\\
0.08 0.865612948071 0.75561427925 0.941497028217 1.75422038426\\
0.085 0.85918096549 0.743807583962 0.93778588462 1.74232913875\\
0.09 0.852855286987 0.732180374524 0.934068055375 1.73061967347\\
0.095 0.846629150994 0.720720796362 0.930343496608 1.71908020579\\
0.1 0.840496506564 0.709418263474 0.926612163925 1.70770020818\\
0.105 0.834451907584 0.698263266606 0.922874012406 1.69647021874\\
0.11 0.828490426965 0.687247218076 0.919128996595 1.68538168788\\
0.115 0.822607586304 0.676362324907 0.915377070489 1.67442685298\\
0.12 0.816799297669 0.665601484088 0.911618187532 1.66359863484\\
0.125 0.811061815012 0.654958195359 0.907852300602 1.65289055157\\
0.13 0.805391693285 0.644426488004 0.904079362002 1.64229664614\\
0.135 0.7997857538 0.634000858941 0.900299323449 1.63181142523\\
0.14 0.794241054679 0.623676220029 0.896512136067 1.62142980713\\
0.145 0.78875486549 0.613447852932 0.892717750371 1.61114707705\\
0.15 0.783324645365 0.603311370237 0.88891611626 1.60095884871\\
0.155 0.777948024001 0.593262681778 0.885107183003 1.59086103101\\
0.16 0.772622785109 0.583297965333 0.881290899231 1.58084979902\\
0.165 0.767346851906 0.573413640994 0.877467212919 1.57092156862\\
0.17 0.762118274368 0.563606348666 0.873636071383 1.56107297422\\
0.175 0.756935217958 0.553872928228 0.869797421258 1.55130084911\\
0.18 0.751795953647 0.544210401964 0.865951208492 1.54160220807\\
0.185 0.746698849033 0.534615958974 0.86209737833 1.53197423189\\
0.19 0.741642360411 0.525086941262 0.858235875302 1.52241425365\\
0.195 0.736625025681 0.51562083131 0.854366643207 1.51291974632\\
0.2 0.731645457976 0.506215240927 0.850489625102 1.50348831177\\
0.205 0.726702339919 0.496867901221 0.846604763287 1.49411767078\\
0.21 0.721794418448 0.487576653555 0.842711999287 1.48480565408\\
0.215 0.716920500124 0.47833944137 0.83881127384 1.47555019425\\
0.22 0.712079446872 0.469154302771 0.834902526883 1.46634931835\\
0.225 0.707270172116 0.460019363791 0.830985697531 1.45720114121\\
0.23 0.702491637248 0.450932832255 0.827060724064 1.44810385938\\
0.235 0.697742848404 0.441892992183 0.823127543911 1.43905574551\\
0.24 0.693022853518 0.432898198669 0.819186093629 1.43005514334\\
0.245 0.688330739609 0.423946873186 0.81523630889 1.42110046293\\
0.25 0.683665630305 0.415037499279 0.811278124459 1.41219017647\\
0.255 0.679026683545 0.406168618602 0.807311474177 1.40332281429\\
0.26 0.674413089476 0.397338827262 0.803336290943 1.39449696123\\
0.265 0.669824068501 0.38854677245 0.799352506689 1.38571125327\\
0.27 0.665258869475 0.379791149321 0.795360052368 1.37696437447\\
0.275 0.660716768036 0.371070698105 0.791358857927 1.36825505405\\
0.28 0.656197065056 0.362384201432 0.787348852285 1.35958206372\\
0.285 0.651699085196 0.353730481835 0.783329963317 1.35094421522\\
0.29 0.647222175566 0.345108399441 0.779302117825 1.34234035801\\
0.295 0.64276570448 0.336516849809 0.775265241518 1.3337693771\\
0.3 0.638329060282 0.327954761914 0.771219258988 1.32523019104\\
0.305 0.633911650259 0.319421096268 0.767164093685 1.31672175006\\
0.31 0.629512899624 0.310914843155 0.763099667888 1.30824303429\\
0.315 0.625132250553 0.302435020983 0.759025902685 1.29979305212\\
0.32 0.620769161295 0.293980674729 0.754942717943 1.29137083866\\
0.325 0.616423105325 0.28555087449 0.750850032279 1.28297545425\\
0.33 0.612093570556 0.277144714111 0.746747763034 1.27460598313\\
0.335 0.607780058587 0.26876130989 0.742635826242 1.26626153212\\
0.34 0.603482084007 0.260399799367 0.738514136599 1.25794122943\\
0.345 0.599199173726 0.252059340172 0.734382607435 1.24964422347\\
0.35 0.594930866348 0.243739108936 0.730241150676 1.24136968179\\
0.355 0.590676711582 0.235438300261 0.726089676819 1.23311679004\\
0.36 0.586436269674 0.227156125746 0.721928094887 1.22488475097\\
0.365 0.582209110879 0.218891813059 0.717756312405 1.21667278355\\
0.37 0.577994814951 0.210644605052 0.713574235355 1.208480122\\
0.375 0.573792970666 0.202413758925 0.709381768141 1.20030601505\\
0.38 0.569603175361 0.194198545425 0.705178813552 1.19214972503\\
0.385 0.565425034498 0.185998248081 0.700965272719 1.18401052718\\
0.39 0.561258161249 0.177812162475 0.696741045073 1.17588770888\\
0.395 0.557102176099 0.169639595543 0.692506028308 1.16778056896\\
0.4 0.552956706463 0.161479864901 0.688260118327 1.15968841698\\
0.405 0.548821386325 0.153332298206 0.684003209205 1.15161057266\\
0.41 0.544695855889 0.145196232536 0.679735193135 1.14354636518\\
0.415 0.540579761239 0.137071013791 0.675455960385 1.1354951326\\
0.42 0.536472754023 0.128955996122 0.671165399241 1.12745622131\\
0.425 0.532374491138 0.120850541372 0.666863395958 1.1194289854\\
0.43 0.528284634436 0.112754018542 0.662549834707 1.11141278619\\
0.435 0.524202850427 0.104665803266 0.658224597514 1.10340699165\\
0.44 0.520128810007 0.0965852773109 0.653887564205 1.09541097589\\
0.445 0.516062188186 0.0885118280767 0.649538612346 1.08742411868\\
0.45 0.512002663823 0.0804448481232 0.645177617176 1.07944580496\\
0.455 0.507949919372 0.0723837346993 0.640804451548 1.07147542431\\
0.46 0.503903640635 0.0643278892857 0.636418985855 1.06351237058\\
0.465 0.499863516521 0.056276717146 0.632021087967 1.05555604131\\
0.47 0.495829238805 0.048229626887 0.627610623152 1.0476058374\\
0.475 0.491800501903 0.0401860300245 0.623187454005 1.03966116257\\
0.48 0.487777002641 0.0321453405574 0.618751440365 1.03172142297\\
0.485 0.483758440037 0.0241069745448 0.61430243924 1.02378602671\\
0.49 0.479744515077 0.0160703496897 0.609840304716 1.01585438348\\
0.495 0.475734930502 0.00803488492458 0.605364887874 1.00792590408\\
0.5 0.471729390599 0 0.600876036693 1\\
0.505 0.467727600982 0 0.596373595961 0.992076082999\\
0.51 0.463729268391 0 0.591857407171 0.984153564687\\
0.515 0.459734100481 0 0.58732730842 0.976231856089\\
0.52 0.455741805615 0 0.5827831343 0.968310367218\\
0.525 0.451752092659 0 0.578224715789 0.960388506649\\
0.53 0.447764670776 0 0.573651880129 0.952465681078\\
0.535 0.443779249219 0 0.569064450707 0.944541294887\\
0.54 0.439795537125 0 0.56446224693 0.9366147497\\
0.545 0.435813243305 0 0.559845084086 0.928685443927\\
0.55 0.431832076035 0 0.555212773212 0.920752772312\\
0.555 0.427851742841 0 0.550565120942 0.912816125459\\
0.56 0.423871950288 0 0.545901929363 0.904874889358\\
0.565 0.41989240376 0 0.541222995848 0.896928444898\\
0.57 0.415912807236 0 0.536528112896 0.888976167364\\
0.575 0.411932863068 0 0.531817067957 0.881017425922\\
0.58 0.407952271748 0 0.527089643247 0.873051583097\\
0.585 0.40397073167 0 0.52234561556 0.865077994223\\
0.59 0.399987938894 0 0.517584756066 0.857096006882\\
0.595 0.39600358689 0 0.512806830103 0.849104960323\\
0.6 0.39201736629 0 0.508011596952 0.841104184863\\
0.605 0.388028964619 0 0.503198809609 0.833093001258\\
0.61 0.384038066027 0 0.49836821454 0.825070720059\\
0.615 0.380044351004 0 0.493519551425 0.817036640933\\
0.62 0.376047496092 0 0.488652552884 0.808990051961\\
0.625 0.372047173579 0 0.483766944201 0.800930228905\\
0.63 0.368043051187 0 0.478862443017 0.792856434436\\
0.635 0.364034791744 0 0.473938759015 0.784767917336\\
0.64 0.36002205284 0 0.468995593589 0.776663911651\\
0.645 0.356004486471 0 0.464032639492 0.768543635807\\
0.65 0.351981738667 0 0.459049580459 0.760406291683\\
0.655 0.347953449095 0 0.454046090818 0.75225106363\\
0.66 0.343919250653 0 0.44902183507 0.744077117439\\
0.665 0.339878769034 0 0.443976467448 0.735883599252\\
0.67 0.335831622272 0 0.438909631448 0.727669634414\\
0.675 0.33177742026 0 0.433820959334 0.719434326251\\
0.68 0.327715764244 0 0.428710071607 0.711176754783\\
0.685 0.323646246288 0 0.423576576443 0.702895975354\\
0.69 0.319568448702 0 0.418420069099 0.694591017176\\
0.695 0.315481943444 0 0.413240131273 0.686260881783\\
0.7 0.311386291475 0 0.408036330427 0.677904541385\\
0.705 0.307281042084 0 0.402808219061 0.669520937104\\
0.71 0.303165732158 0 0.397555333942 0.661108977103\\
0.715 0.299039885411 0 0.392277195272 0.652667534574\\
0.72 0.294903011558 0 0.386973305804 0.64419544559\\
0.725 0.290754605428 0 0.381643149889 0.635691506797\\
0.73 0.286594146023 0 0.376286192461 0.627154472935\\
0.735 0.282421095492 0 0.370901877931 0.618583054178\\
0.74 0.278234898049 0 0.365489629015 0.609975913261\\
0.745 0.27403497879 0 0.360048845459 0.601331662384\\
0.75 0.269820742425 0 0.354578902665 0.592648859867\\
0.755 0.265591571913 0 0.349079150208 0.583926006527\\
0.76 0.261346826977 0 0.343548910233 0.575161541753\\
0.765 0.257085842497 0 0.337987475708 0.566353839245\\
0.77 0.252807926771 0 0.332394108541 0.557501202373\\
0.775 0.248512359614 0 0.326768037521 0.548601859123\\
0.78 0.244198390297 0 0.321108456075 0.539653956577\\
0.785 0.239865235285 0 0.315414519825 0.530655554875\\
0.79 0.235512075772 0 0.309685343915 0.521604620594\\
0.795 0.231138054971 0 0.303920000072 0.512499019472\\
0.8 0.226742275142 0 0.298117513395 0.503336508396\\
0.805 0.22232379431 0 0.2922768588 0.494114726564\\
0.81 0.217881622651 0 0.286396957116 0.4848311857\\
0.815 0.213414718485 0 0.280476670756 0.475483259204\\
0.82 0.208921983827 0 0.27451479892 0.466068170085\\
0.825 0.204402259442 0 0.268510072257 0.456582977498\\
0.83 0.199854319319 0 0.26246114692 0.447024561692\\
0.835 0.195276864487 0 0.25636659791 0.437389607113\\
0.84 0.190668516065 0 0.250224911611 0.427674583388\\
0.845 0.186027807432 0 0.244034477382 0.417875723834\\
0.85 0.181353175367 0 0.237793578058 0.407989001103\\
0.855 0.176642949983 0 0.231500379179 0.398010099447\\
0.86 0.171895343253 0 0.225152916717 0.387934383023\\
0.865 0.167108435871 0 0.218749083043 0.377756859501\\
0.87 0.162280162121 0 0.2122866108 0.367472138088\\
0.875 0.15740829239 0 0.205763054276 0.357074380857\\
0.88 0.152490412829 0 0.199175767779 0.34655724601\\
0.885 0.14752390157 0 0.192521880389 0.33591382135\\
0.89 0.142505900734 0 0.185798266282 0.325136545775\\
0.895 0.137433283273 0 0.179001509633 0.314217115999\\
0.9 0.132302613384 0 0.172127862784 0.303146374895\\
0.905 0.127110098886 0 0.165173195997 0.291914176732\\
0.91 0.121851533411 0 0.15813293656 0.280509223086\\
0.915 0.116522225549 0 0.151001994273 0.268918861051\\
0.92 0.111116911054 0 0.143774669286 0.257128832382\\
0.925 0.105629642778 0 0.13644453674 0.245122957844\\
0.93 0.100053650797 0 0.129004300441 0.232882734584\\
0.935 0.0943811619592 0 0.121445604408 0.22038681469\\
0.94 0.0886031630599 0 0.113758786015 0.207610317983\\
0.945 0.0827090838406 0 0.105932546156 0.194523908146\\
0.95 0.0766863628629 0 0.0979534984195 0.181092521625\\
0.955 0.0705198368345 0 0.0898055361458 0.167273570714\\
0.96 0.0641908537032 0 0.0814689150144 0.153014319904\\
0.965 0.057675933708 0 0.0729188708759 0.13824790209\\
0.97 0.0509446485535 0 0.0641234350979 0.122886968586\\
0.975 0.0439560499956 0 0.055039763728 0.106812921072\\
0.98 0.0366521475437 0 0.0456074489974 0.0898560913475\\
0.985 0.0289445676163 0 0.0357348717764 0.0717548305326\\
0.99 0.020682137338 0 0.0252661277271 0.0520550037332\\
0.995 0.0115438588076 0 0.0138721635339 0.0297732646289\\
1 0 0 0 0\\
}\adRateData

\begin{tikzpicture}
\begin{axis}[
 width=0.96\linewidth,height=8cm,
 xmin=0,xmax=1,ymin=0,ymax=2.04,
 xlabel={Damping probability $p$},
 ylabel={Numerical rate (units in caption)},
 xtick={0,0.1,0.2,0.3,0.4,0.5,0.6,0.7,0.8,0.9,1},
 ytick={0,0.25,0.5,0.75,1,1.25,1.5,1.75,2},
 grid=major,grid style={opacity=0.18},
 tick label style={font=\small},
 legend style={font=\footnotesize,at={(0.98,0.98)},anchor=north east,
               draw=none,fill=white},
 legend cell align=left]
\addplot+[thick,mark=none,solid] table[x=p,y=C]{\adRateData};
\addlegendentry{$C(\cA_p)=\chi(\cA_p)$}
\addplot+[thick,mark=none,dashed] table[x=p,y=Q]{\adRateData};
\addlegendentry{$Q(\cA_p)$}
\addplot+[thick,mark=none,dashdotted] table[x=p,y=cost]{\adRateData};
\addlegendentry{$\Epar(\cA_p)$}
\addplot+[thick,mark=none,dotted] table[x=p,y=CE]{\adRateData};
\addlegendentry{$C_E(\cA_p)$}
\node[font=\small,anchor=south] at (axis cs:0.70,0.06)
 {$Q=0$ for $p\ge1/2$};
\end{axis}
\end{tikzpicture}
\caption{Operational rates of the qubit amplitude-damping channel.
The solid curve is the unassisted classical capacity
(\cref{thm:damping-capacity}), evaluated from the single-use
expression of \cite{giovannetti-fazio}. The dash-dotted curve is
$\Epar(\cA_p)$, equal to $h_2((1+\sqrt p)/2)$ and also
to the channel entanglement of formation and the Choi-state
entanglement of formation and cost (\cref{thm:ad-cost}).
The dashed and dotted curves give the established quantum and
entanglement-assisted classical capacities
\cite{khatri,op:BSST,op:DS}; the latter also equals the
entanglement-assisted classical simulation cost specified in
\cref{app:operational}. Rates are measured in bits per use
for $C,C_E$, qubits per use for $Q$, ebits per use for
$\Epar$, and ebits per copy for $\EC(J_{\cA_p})$.
The resource conventions differ between curves.}
\label{fig:ad-rates}
\end{figure}

\FloatBarrier

\section{Strong superadditivity of entanglement of formation}\label{sec:support-proof}
The channel arguments rely on a state-level inequality: a
constraint on one marginal prevents arbitrary auxiliary
correlations from reducing the sum of the two formation costs.
We first state this support criterion, then identify the two
estimates used in its proof.

Choose unit vectors $\ket{0_A},\ket{0_B}$ and decompose
$A=\mathbb C\ket{0_A}\oplus A_1$ and
$B=\mathbb C\ket{0_B}\oplus B_1$. We write 
\begin{equation}\label{eq:allowed-support}
 \cS(A,B)=\mathbb C\ket{00}
       \oplus(A_1\otimes\ket{0_B})
       \oplus(\ket{0_A}\otimes B_1)
\end{equation}
as the orthogonal complement of $A_1\otimes B_1$.
The support constraint allows arbitrary coherence between its
three summands and imposes no superselection rule.

\begin{theorem}[Strong superadditivity of entanglement of formation]
\label[theorem]{thm:support}
Let $A,B,A',B'$ be arbitrary finite-dimensional systems. If
$\Omega_{ABA'B'}$ has marginal
$\rho_{AB}=\Tr_{A'B'}\Omega$ with
$\supp\rho\subseteq\cS(A,B)$, then, writing
$\sigma_{A'B'}=\Tr_{AB}\Omega$,
\begin{equation}\label{eq:strong-state}
 \EF(\Omega_{AA':BB'})\ge
       \EF(\rho_{A:B})+\EF(\sigma_{A':B'}).
\end{equation}
In particular, for every such $\rho$ and every finite-dimensional
bipartite state $\sigma$,
\begin{equation}\label{eq:strong-state-add}
 \EF(\rho\otimes\sigma)=\EF(\rho)+\EF(\sigma),
 \qquad \EC(\rho)=\EF(\rho).
\end{equation}
\end{theorem}

Only the marginal $\rho_{AB}$ is constrained. The auxiliary
systems have arbitrary finite dimensions, and $\Omega$ may
be mixed and correlated. For two qubits, the excluded sector
is one product direction. In higher dimensions it is the
entire subspace $A_1\otimes B_1$, not merely a single vector.

The proof first treats a pure extension. Its missing joint
sector produces a triangular coefficient matrix. The scalar
term in the matrix entropy bound will dominate $\EF(\rho)$,
while the weighted block entropies will dominate $\EF(\sigma)$.
These two comparisons require only upper bounds on the
marginal formation costs, not exact convex-roof formulas.
Convexity then extends the result to mixed extensions.

\paragraph{Entropy bound for triangular coefficient matrices.}
For a bipartite vector $\ket v=\sum_{ij}M_{ij}\ket i\ket j$,
the coefficient matrix obeys
$\langle v|v\rangle=\normF M^2$ and
$\Tr_2\proj v=MM^\dagger$. If $v$ is normalized, its
entanglement entropy is therefore $S(MM^\dagger)$.

\Needspace{14\baselineskip}
\begin{theorem}[Triangular block entropy inequality]\label[theorem]{thm:triangle}
Let $A\in\mathbb C^{r_1\times c_1}$,
$B\in\mathbb C^{r_2\times c_1}$, and
$D\in\mathbb C^{r_2\times c_2}$, with arbitrary positive integer block
dimensions. Set
\begin{equation}\label{eq:triangle-setup}
 T=\begin{pmatrix}A&0\\B&D\end{pmatrix},\qquad
 \alpha=\normF A^2,\quad\beta=\normF B^2,\quad\gamma=\normF D^2,
 \qquad\alpha+\beta+\gamma=1.
\end{equation}
Then
\begin{equation}\label{eq:triangle-bound}
 S(TT^\dagger)\ge{}e(2\sqrt{\alpha\gamma})
       +\alpha S(AA^\dagger/\alpha)
       +\beta S(BB^\dagger/\beta)
       +\gamma S(DD^\dagger/\gamma).
\end{equation}
A weighted entropy is zero when its weight vanishes.
\end{theorem}
No relation between the singular vectors of the blocks is
assumed. Replacing each block with its Frobenius norm gives
the scalar coefficient matrix
\begin{equation}\label{eq:scalar-coefficient}
 T_{\mathrm{sc}}=
 \begin{pmatrix}\sqrt\alpha&0\\\sqrt\beta&\sqrt\gamma\end{pmatrix},
 \qquad S(T_{\mathrm{sc}}T_{\mathrm{sc}}^\dagger)
     =e(2\sqrt{\alpha\gamma}).
\end{equation}
Thus \eqref{eq:triangle-bound} separates the entropy of $T$
into a scalar contribution and weighted within-block entropies.
The scalar contribution retains the coherence between the
three sectors; it is not their Shannon entropy. The bound
is attained by $A=\sqrt\alpha R$, $B=\sqrt\beta R$, and
$D=\sqrt\gamma R$ for any $R$ with $\normF R=1$, since
then $T=T_{\mathrm{sc}}\otimes R$.

The proof is given in \cref{app:triangle}. Polar decomposition
reduces the problem to an entropy difference $\cF(P,Q,R)$.
A log-determinant representation establishes convexity in
$(P,R)$ at fixed $Q$; dephasing in an eigenbasis of $Q$
and scalar convexity then reduce the bound to the three
traces. \Cref{prop:triangle-equality} gives the equality
conditions for square blocks of the same size when $B$
is invertible.

\paragraph{A marginal bound from fixed-expectation decompositions.}
To compare the scalar term with the first marginal, we need
a decomposition with the same two sector weights in every
component. The following result provides it
\cite{leka-petz,leka}. We include a proof by rank reduction.

\begin{lemma}[Pure-state decomposition preserving two expectations]
\label[lemma]{lem:two-moments}
Let $\rho$ be a finite-dimensional state and let $F,G$ be Hermitian operators.
Then there exist normalized vectors $\ket{\phi_j}\in\supp\rho$ and probabilities
$w_j>0$, $\sum_j w_j=1$, such that
\[
    \rho=\sum_j w_j\proj{\phi_j},
\]
and
\begin{equation}
\label{eq:fixed-moments}
    \bra{\phi_j}F\ket{\phi_j}=\Tr(F\rho),
    \qquad
    \bra{\phi_j}G\ket{\phi_j}=\Tr(G\rho)
\end{equation}
for every $j$.
\end{lemma}

\begin{proof}
We proceed by induction on $r=\rank\rho$. The case $r=1$ is immediate. Assume $r=\rank\rho\ge2$.  Choose a basis in which
\[
    \rho=
    \begin{pmatrix}
        \rho_r & 0\\
        0 & 0
    \end{pmatrix},
    \qquad
    \rho_r>0,
\]
where $\rho_r$ is an $r\times r$ matrix.  We seek a Hermitian
perturbation of the form
\[
    H=
    \begin{pmatrix}
        H_r & 0\\
        0 & 0
    \end{pmatrix},
\]
so that $\rho+tH$ remains within the support of $\rho$.

An $r\times r$ Hermitian matrix $H_r$ has $r^2$ real degrees of
freedom.  Requiring $\rho+tH$ to preserve its trace and its
$F$- and $G$-expectation values imposes the three homogeneous linear
conditions
\begin{equation}\label{eq:two-moment-H}
    \Tr H=0,\qquad
    \Tr(FH)=0,\qquad
    \Tr(GH)=0.
\end{equation}
Since $r\ge2$, we have $r^2>3$.  Hence these three linear conditions
cannot force $H_r=0$: there exists a nonzero Hermitian $H$ satisfying
all three. Consider
\[
    I_H:=\{t\in\mathbb R:\rho+tH\succeq0\}.
\]
Since $\rho$ is strictly positive on its support, $I_H$ contains an
open neighborhood of $0$. Moreover, $I_H$ is closed and convex, hence
an interval. Because $H\neq0$ and $\Tr H=0$, the Hermitian operator
$H$ has both a positive and a negative eigenvalue. Therefore
$\rho+tH$ fails to be positive semidefinite for sufficiently large
positive $t$ and also for sufficiently large negative $t$. Thus we can write $I_H$ as a closed set $I_H=[t_-,t_+]$ for finite numbers $t_-<0<t_+$.

Define
\[
    \rho_-:=\rho+t_-H,
    \qquad
    \rho_+:=\rho+t_+H.
\]
By \eqref{eq:two-moment-H}, both are normalized states and preserve the
two expectation values:
\[
    \Tr(F\rho_\pm)=\Tr(F\rho),
    \qquad
    \Tr(G\rho_\pm)=\Tr(G\rho).
\]
They are also supported on $\supp\rho$. Since $t_-$ and $t_+$ are
boundary points of $I_H$, neither $\rho_-$ nor $\rho_+$ can be strictly
positive on $\supp\rho$; otherwise positivity would persist beyond the
corresponding endpoint. Hence
\[
    \rank\rho_-<r,
    \qquad
    \rank\rho_+<r.
\]
Finally,
\begin{equation}
\label{eq:rank-reduction-decomposition}
    \rho
    =
    \frac{t_+}{t_+-t_-}\rho_-
    +
    \frac{-t_-}{t_+-t_-}\rho_+ .
\end{equation}
The two coefficients are strictly positive and sum to one. By the
induction hypothesis, each of $\rho_-$ and $\rho_+$ admits a finite
pure-state decomposition inside $\supp\rho$ whose components preserve
the same $F$- and $G$-expectation values. Combining these two
decompositions using \eqref{eq:rank-reduction-decomposition} gives the
desired decomposition of $\rho$.
\end{proof}

The restriction to two expectations is essential in general:
the three Pauli expectations determine a qubit state, so
a mixed qubit state has no pure-state decomposition preserving
all three. Our application requires only two support projectors.

To apply the lemma, consider a normalized pure state in
$\cS(A,B)$. It has the form
\begin{equation}\label{eq:single-pure}
 \ket\phi=c\ket{00}+\ket a\ket0+\ket0\ket d,
 \qquad a\in A_1,\ d\in B_1,
 \quad\alpha=\|a\|^2,\ \gamma=\|d\|^2.
\end{equation}
Its local supports have dimension at most two. In suitable
bases, its coefficient matrix is
$\left(\begin{smallmatrix}c&\sqrt\gamma\\
\sqrt\alpha&0\end{smallmatrix}\right)$, so its two Schmidt
probabilities have sum one and product $\alpha\gamma$.
This includes a zero probability when the Schmidt rank is one.
Consequently,
\begin{equation}\label{eq:pure-support-entropy}
 S(\phi_A)=e(2\sqrt{\alpha\gamma}).
\end{equation}
For a mixed state $\rho$ supported on $\cS(A,B)$, apply
\cref{lem:two-moments} with $F=P_{A_1}\otimes I$ and
$G=I\otimes P_{B_1}$. Every component then has the same
$\alpha,\gamma$ and hence the same entropy in
\eqref{eq:pure-support-entropy}. It follows that
\begin{equation}\label{eq:marginal-upper}
 \EF(\rho)\le e(2\sqrt{\alpha\gamma}),\qquad
 \alpha=\Tr(F\rho),\quad\gamma=\Tr(G\rho).
\end{equation}
This upper bound need not be tight: for
$\rho=\tfrac12\proj{10}+\tfrac12\proj{01}$, its two
sides are zero and one, respectively. What matters is that
the bound uses precisely the weights in
\eqref{eq:triangle-bound}.

\begin{proof}[Proof of \cref{thm:support}]
If $A_1$ or $B_1$ is zero, $\rho$ is separable and
\eqref{eq:strong-state} follows from monotonicity under
local discarding. We may therefore assume both are nonzero.

\emph{Pure extensions.} Suppose $\Omega=\proj\Psi$.
Positivity and the support hypothesis give
$(P_{A_1}\otimes P_{B_1}\otimes I)\ket\Psi=0$.
Across the bipartition $AA':BB'$, order rows by
$A_1A',0_AA'$ and columns by $0_BB',B_1B'$. The
coefficient matrix is then
\begin{equation}\label{eq:extension-blocks}
 T=\begin{pmatrix}A&0\\B&D\end{pmatrix},\qquad
 \begin{array}{ll}
 A:&(A_1A'):B',\\
 B:&A':B',\\
 D:&A':(B_1B').
 \end{array}
\end{equation}
The squared block norms $\alpha,\beta,\gamma$ are those
of \eqref{eq:triangle-setup}. In particular,
$\alpha,\gamma$ are the expectations of $F,G$ in $\rho$.
Equation \eqref{eq:marginal-upper} thus bounds $\EF(\rho)$
by the scalar term in \eqref{eq:triangle-bound}.

To bound $\EF(\sigma)$, normalize each nonzero branch,
discard $A_1$ from the first branch and $B_1$ from the
third, and average with weights $\alpha,\beta,\gamma$.
This yields $\sigma_{A'B'}$: tracing out $AB$ removes
all coherences between distinct sectors. Local monotonicity
of entanglement of formation bounds each reduced branch
by its original pure-state entropy. Convexity then gives
\begin{equation}\label{eq:auxiliary-bound}
 \EF(\sigma)\le
 \alpha S(AA^\dagger/\alpha)+\beta S(BB^\dagger/\beta)
                         +\gamma S(DD^\dagger/\gamma).
\end{equation}
Combining the two marginal bounds with
\cref{thm:triangle} proves
$S(\Psi_{AA'})\ge\EF(\rho)+\EF(\sigma)$.

\emph{Mixed extensions.} Let
$\Omega=\sum_jw_j\proj{\Psi_j}$ be any pure-state
decomposition. The support hypothesis implies
\[
 0=\Tr[(P_{A_1}\otimes P_{B_1}\otimes I)\Omega]
   =\sum_jw_j\|(P_{A_1}\otimes P_{B_1}\otimes I)\Psi_j\|^2.
\]
Each positive-weight component therefore satisfies the same
constraint. Writing its marginals as $\rho_j,\sigma_j$,
the pure-extension bound and convexity yield
\[
 \sum_jw_jS((\Psi_j)_{AA'})
 \ge\sum_jw_j\EF(\rho_j)+\sum_jw_j\EF(\sigma_j)
 \ge\EF(\rho)+\EF(\sigma).
\]
Minimizing over decompositions of $\Omega$ proves
\eqref{eq:strong-state}.

\emph{Additivity and state cost.} For $\Omega=\rho\otimes\sigma$,
the preceding inequality gives the lower bound in
\eqref{eq:strong-state-add}; products of optimal decompositions
give the upper bound. Iteration yields
$\EF(\rho^{\otimes n})=n\EF(\rho)$ for every $n\ge1$.
The regularization theorem \eqref{eq:state-cost} then gives
$\EC(\rho)=\EF(\rho)$.
\end{proof}
Only upper bounds on the marginal convex roofs enter the
argument. Neither the fixed-expectation decomposition nor
the reduced branch decompositions need be optimal.

\begin{corollary}[Two-qubit states with a product vector in the kernel]
\label[corollary]{cor:state-examples}
For two qubits, \eqref{eq:strong-state} and
\eqref{eq:strong-state-add} hold whenever $\ker\rho$ contains a nonzero
product vector. Every two-qubit state of rank at most two satisfies
this condition.
\end{corollary}
\begin{proof}
A local unitary takes a product kernel vector to $\ket{11}$,
whose orthogonal complement is $\cS(A,B)$. The first
claim is therefore an instance of \cref{thm:support}.

If $\rank\rho\le2$, choose two independent kernel
vectors with coefficient matrices $L_0,L_1\in\mathbb C^{2\times2}$.
The homogeneous polynomial $\det(uL_0+vL_1)$ either
vanishes identically or has a zero on the complex
projective line. At such a nonzero pair $(u,v)$,
linear independence ensures $uL_0+vL_1\ne0$.
The resulting matrix is singular and has rank one,
so it represents a product vector in $\ker\rho$.
\end{proof}
In dimensions $d_A,d_B$, the allowed subspace has dimension
$d_A+d_B-1$; for example, it accommodates rank-five states
in $3\times3$. The higher-dimensional condition excludes
the whole subspace $A_1\otimes B_1$, of dimension
$(d_A-1)(d_B-1)$. If several marginals satisfy the
hypothesis, iteration gives
\[
 \EF(\Omega_{A_1\cdots A_m:B_1\cdots B_m})
       \ge\sum_{i=1}^m\EF(\Omega_{A_i:B_i})
\]
for any correlated global state, allowing a different pair
of distinguished local vectors for each marginal.

For two-qubit calculations, we shall also use
\begin{equation}\label{eq:qubit-kernel}
 \rho\ket{01}=0\quad\Longrightarrow\quad
 \EF(\rho)=e(2|\bra{00}\rho\ket{11}|).
\end{equation}
This identity follows from Wootters' formula
\cite{wootters}; a derivation is given in \cref{app:qubit}.
It is used only for explicit evaluations, not in the
support theorem.

The two channel applications use the support theorem across
different bipartitions: \cref{sec:capacity} applies it to
output--environment states for every input, whereas
\cref{sec:simulation} applies it to the reference--output
Choi state. For a qubit-to-qubit channel $\cN$ and a
Stinespring isometry $V:X\to B\otimes E$, the corresponding
conditions are equivalent:
\begin{equation}\label{eq:pure-output-equivalence}
 \begin{aligned}
  \cN\text{ admits a pure output}
  &\iff \ran V\subseteq\cS(B,E)
       \text{ for suitable distinguished vectors}\\
  &\iff \ker J_{\cN}\text{ contains a nonzero product vector}.
 \end{aligned}
\end{equation}
The first equivalence is proved in
\cref{lem:pure-output-stinespring}; the second follows from
\eqref{eq:choi-product}. Thus the two applications cover the
same pure-output class, which includes every qubit-to-qubit
channel of Kraus rank at most two.

\section{Classical capacity and finite-block optimality}\label{sec:capacity}
We now prove the communication statements previewed in
\cref{sec:results}. The argument has three steps. The
Stinespring correspondence converts \cref{thm:support}
into a lower bound on the output-entropy convex roof.
Entropy subadditivity then gives a converse for correlated
inputs. For amplitude damping, a fixed-input calculation
and an explicit binary ensemble evaluate and attain the
resulting capacity. We conclude with a finite-block
refinement and the optimal average input.

\subsection{Strong Holevo additivity from Stinespring support}
Let $V:X\to B\otimes E$ be a Stinespring isometry for
$\cN$. The Matsumoto--Shimono--Winter identity
\cite{msw} reads
\begin{equation}\label{eq:msw}
 \roof{\cN}(\tau)=\EF((V\tau V^\dagger)_{B:E}).
\end{equation}
Indeed, $V$ gives a bijection between pure-state
decompositions of $\tau$ and $V\tau V^\dagger$.
Every component of the latter lies in $\ran V$ and
pulls back under $V^\dagger$ to a normalized component
of $\tau$, with the same output entropy.

\begin{corollary}[Stinespring criterion for strong Holevo additivity]\label[corollary]{cor:stinespring}
Suppose $\cN:X\to B$ has a dilation $V:X\to B\otimes E$ with
$\ran V\subseteq\cS(B,E)$. For every finite-dimensional partner
$\cM:Y\to B'$ and every joint input $\tau_{XY}$,
\begin{equation}\label{eq:roof-superadd}
 \roof{\cN\otimes\cM}(\tau_{XY})
 \ge\roof{\cN}(\tau_X)+\roof{\cM}(\tau_Y).
\end{equation}
Consequently,
\begin{equation}\label{eq:stinespring-add}
 \chi(\cN\otimes\cM)=\chi(\cN)+\chi(\cM),\qquad
 C(\cN)=\chi(\cN),\qquad
 C(\cN\otimes\cM)=C(\cN)+C(\cM).
\end{equation}
\end{corollary}
\begin{proof}
Choose a Stinespring isometry $W:Y\to B'\otimes E'$
for $\cM$. The $BE$ marginal of
$(V\otimes W)\tau(V\otimes W)^\dagger$ is supported
on $\cS(B,E)$. Applying \cref{thm:support} across
$BB':EE'$ and then \eqref{eq:msw} proves
\eqref{eq:roof-superadd}. Subadditivity of output
entropy now gives
\begin{equation}\label{eq:fixed-holevo-converse}
 \chi_{\cN\otimes\cM}(\tau)
 \le\chi_{\cN}(\tau_X)+\chi_{\cM}(\tau_Y)
 \le\chi(\cN)+\chi(\cM).
\end{equation}
Maximize over $\tau$ and use independent optimal
ensembles for the reverse inequality. This proves
strong Holevo additivity. Iteration and \eqref{eq:hsw}
give $C(\cN)=\chi(\cN)$. To obtain additivity of
the capacities, apply the Holevo identity successively
to the $\cN$ factors in
$\cN^{\otimes n}\otimes\cM^{\otimes n}$ and
regularize. This step requires no additivity assumption
on $\cM$.
\end{proof}

\begin{lemma}[Pure outputs and Stinespring support]
\label[lemma]{lem:pure-output-stinespring}
Let $\cN$ be a qubit-to-qubit channel, and let
$V:X\to B\otimes E$ be a Stinespring isometry.
Then $\cN$ admits a pure output if and only if
$\ran V\subseteq\cS(B,E)$ for suitable distinguished unit
vectors in $B$ and $E$. Every such channel therefore satisfies
\eqref{eq:roof-superadd} and \eqref{eq:stinespring-add}.
\end{lemma}
\begin{proof}
Suppose $\cN(\proj{\psi_0})=\proj{b_0}$.
Purity of the reduced output implies
$V\ket{\psi_0}=\ket{b_0}\ket u$ for a unit vector $u\in E$.
Complete $\psi_0$ and $b_0$ to orthonormal input and output
bases, respectively. Since $B$ is two-dimensional, there
are vectors $v,w\in E$ such that
\[
 V\ket{\psi_1}=\ket{b_0}\ket v+\ket{b_1}\ket w.
\]
If $w\ne0$, choose $e_0=w/\|w\|$; if $w=0$, choose any
unit vector $e_0\in E$. In either case, linearity gives
\[
 \ran V\subseteq
 (\ket{b_0}\otimes E)\oplus\operatorname{span}\{\ket{b_1e_0}\}
 =\cS(B,E).
\]
No restriction on the finite dimension of $E$ is needed.

Conversely, under this support condition, write
\[
 V\ket\psi=\ket{b_0}\ket{x_\psi}
                  +\ell(\psi)\ket{b_1e_0},
 \qquad
 \ell(\psi)=\bra{b_1e_0}V\ket\psi.
\]
The scalar linear functional $\ell:X\to\mathbb C$ has a
nonzero kernel because $\dim X=2$. A unit vector $\psi$
in this kernel has a product image under $V$, and hence
$\cN(\proj\psi)=\proj{b_0}$. The communication conclusions
now follow from \cref{cor:stinespring}.
\end{proof}

\begin{corollary}[Qubit channels with at most two Kraus operators]
\label[corollary]{cor:two-kraus}
Every qubit-to-qubit channel of Kraus rank at most two admits
a pure output and satisfies \eqref{eq:roof-superadd} and
\eqref{eq:stinespring-add}.
\end{corollary}
\begin{proof}
Such a channel admits a qubit environment, embedding
a one-dimensional environment if necessary. Its
Stinespring image is two-dimensional, as is its
orthogonal complement in $B\otimes E$. The determinant
argument from \cref{cor:state-examples} gives a product
vector in this complement. After local unitaries, the image lies in $\cS(B,E)$.
An environment unitary leaves the channel unchanged, and an
output unitary preserves both output entropy and Holevo
information, including with a partner channel. Hence
\cref{lem:pure-output-stinespring,cor:stinespring} apply
also to the original channel.
\end{proof}

\subsection{Classical capacity of amplitude damping}
The Stinespring criterion removes the many-use optimization.
It remains to evaluate the single-use expression. At a fixed
average input, coherence can affect the output entropy,
but the output-entropy convex roof depends only on the
excitation probability. For an input with excitation probability $q$, define
\begin{align}
 g_p(q)&=e\!\left(2\sqrt{p(1-p)}\,q\right),\label{eq:g}\\
 F_p(q)&=h_2((1-p)q)-g_p(q).\label{eq:F}
\end{align}
The function $g_p(q)$ is the output entropy of a pure input
with excitation $q$, whereas $h_2((1-p)q)$ is the output
entropy of the diagonal input with the same excitation.

\begin{lemma}[Fixed-input output-entropy roof]\label[lemma]{lem:flat-roof}
For $p,q\in[0,1]$ and $|z|^2\le q(1-q)$, let
$\tau=\left(\begin{smallmatrix}1-q&z\\z^*&q\end{smallmatrix}\right)$.
Then
\begin{equation}\label{eq:flat-roof}
 \roof{\cA_p}(\tau)=g_p(q).
\end{equation}
\end{lemma}
\begin{proof}
A pure input with excitation $q$ has output determinant
$p(1-p)q^2$ and hence output entropy $g_p(q)$.
Every pure-state decomposition of $\tau$ has mean
excitation $q$. Since $g_p$ is convex, Jensen's
inequality gives $\roof{\cA_p}(\tau)\ge g_p(q)$.
For the reverse inequality, suppose $0<q<1$ and write
$z=r\sqrt{q(1-q)}e^{i\theta}$, with $0\le r\le1$
and arbitrary $\theta$ if $z=0$. The states
$\ket{\phi_\pm}=\sqrt{1-q}\ket0
\pm e^{-i\theta}\sqrt q\ket1$, taken with probabilities
$(1\pm r)/2$, average to $\tau$. Both have excitation
$q$ and output entropy $g_p(q)$. Thus this explicit
fixed-excitation decomposition attains the lower bound.
The cases $q=0,1$ are pure inputs; the same identity also
holds at $p=0,1$, when $g_p$ vanishes.
\end{proof}
Dephasing the average output cannot decrease its entropy, so
\begin{equation}\label{eq:single-output-entropy}
 S(\cA_p(\tau))\le h_2((1-p)q).
\end{equation}
Together with \cref{lem:flat-roof}, this yields
$\chi_{\cA_p}(\tau)\le F_p(q)$. The following theorem
combines this bound with strong additivity and the binary
ensemble attaining equality.

\begin{theorem}[Classical capacity of amplitude damping]
\label[theorem]{thm:damping-capacity}
For every $p\in[0,1]$,
\begin{equation}\label{eq:main-capacity}
 \begin{split}
 C(\cA_p)=\chi(\cA_p)
 &=\max_{0\le q\le1}\left\{
 h_2((1-p)q)
 -g_p(q)
 \right\}.
 \end{split}
\end{equation}
For any maximizing $q_*$, the equiprobable pure states
\begin{equation}\label{eq:optimal-signals}
 \ket{\psi_\pm(q_*)}=\sqrt{1-q_*}\ket0\pm\sqrt{q_*}\ket1
\end{equation}
attain the single-use optimum. Product codewords from this ensemble,
together with collective decoding, achieve every rate below
$C(\cA_p)$. The arbitrary-partner additivity identities in
\eqref{eq:stinespring-add} hold with $\cN=\cA_p$.
\end{theorem}
\begin{proof}
\emph{Strong additivity and the converse.}
The dilation \eqref{eq:dilation} satisfies
$\ran V_p\subseteq\cS(B,E)$. Hence \cref{cor:stinespring}
applies: both the fixed-input bound \eqref{eq:roof-superadd}
and the channel identities \eqref{eq:stinespring-add}
hold with $\cN=\cA_p$. In particular,
$C(\cA_p)=\chi(\cA_p)$. By \cref{lem:flat-roof} and
\eqref{eq:single-output-entropy}, this capacity is at most
$\max_{0\le q\le1}F_p(q)$.

\emph{Evaluation and achievability.}
For a fixed $q$, the equiprobable signals
$\psi_\pm(q)$ in \eqref{eq:optimal-signals} have average
\begin{equation}\label{eq:tauq}
 \tau_q=(1-q)\proj0+q\proj1.
\end{equation}
Their individual output entropies are $g_p(q)$, and the
entropy of their average output is $h_2((1-p)q)$.
Their Holevo information is therefore exactly $F_p(q)$.
Choosing a maximizing $q_*$ proves \eqref{eq:main-capacity}.
The classical--quantum coding theorem supplies product
codewords and collective decoders achieving every smaller
rate \cite{holevo,schumacher-westmoreland}.
At $p=0,1$, the channel is respectively the identity or a
constant channel, giving $C(\cA_0)=1$ and $C(\cA_1)=0$.
\end{proof}
The scalar optimization and optimal binary ensemble are
those of \cite[Sec.~III.B]{op:BSST} and
\cite[Eqs.~(42)--(43)]{giovannetti-fazio}. Strong Holevo
additivity supplies the missing many-use converse and gives
$\chi(\cA_p^{\otimes n})=n\chi(\cA_p)$ at every blocklength.

\subsection{Finite-block converse and the optimal average input}
For $\tau\in\cD(X^{\otimes n})$, write
$q_i=\bra1\tau_{X_i}\ket1$. Iterating
\eqref{eq:roof-superadd} with $\cN=\cA_p$, using
\cref{lem:flat-roof}, and applying output entropy
subadditivity gives
\begin{equation}\label{eq:finite-block}
 \chi_{\cA_p^{\otimes n}}(\tau)\le\sum_{i=1}^nF_p(q_i).
\end{equation}
This bound holds for arbitrary correlations in $\tau$.
For $0\le p<1$, the function $q\mapsto h_2((1-p)q)$
is strictly concave and $g_p$ is convex. Thus $F_p$ is
strictly concave on $(0,1)$; since $F_p(0)=F_p(1)=0$,
its maximizer $q_*\in(0,1)$ is unique.

The finite-block converse admits an exact decomposition
into three nonnegative terms. Fix $0\le p<1$ and set
\begin{equation}\label{eq:deficit-setup}
 \tau_*=\tau_{q_*},\qquad\sigma_*=\cA_p(\tau_*),\qquad
 b_p(q)=g_p(q)-g_p(q_*)-g_p'(q_*)(q-q_*).
\end{equation}
The state $\sigma_*$ is full rank. Convexity gives
$b_p(q)\ge0$, while stationarity at the interior
maximizer gives $(1-p)h_2'((1-p)q_*)=g_p'(q_*)$.

\begin{proposition}[Finite-block Holevo deficit]\label[proposition]{prop:deficit}
Let $0\le p<1$, $\tau\in\cD(X^{\otimes n})$, $\sigma=\cA_p^{\otimes n}(\tau)$,
and $q_i=\bra1\tau_{X_i}\ket1$. Set
\begin{equation}\label{eq:roof-excess}
 R_n(\tau)=\roof{\cA_p^{\otimes n}}(\tau)-\sum_{i=1}^n g_p(q_i)\ge0.
\end{equation}
Then
\begin{equation}\label{eq:exact-deficit}
 nC(\cA_p)-\chi_{\cA_p^{\otimes n}}(\tau)
   =D(\sigma\Vert\sigma_*^{\otimes n})+R_n(\tau)+\sum_{i=1}^n b_p(q_i).
\end{equation}
In particular, for every input $\rho$,
\begin{equation}\label{eq:radius-bound}
 D\!\left(\cA_p^{\otimes n}(\rho)\middle\Vert\sigma_*^{\otimes n}\right)
       \le nC(\cA_p).
\end{equation}
The unique average input attaining $nC(\cA_p)$ is
$\tau_*^{\otimes n}$.
\end{proposition}
\begin{proof}
Equation \eqref{eq:roof-superadd}, applied successively
with $\cN=\cA_p$, and \eqref{eq:flat-roof} imply
$R_n(\tau)\ge0$.
The eigenvalues of $\sigma_*$ are $1-(1-p)q_*$ and
$(1-p)q_*$. Its logarithm is therefore affine in
the excitation projector, which gives
\begin{equation}\label{eq:relative-entropy-expansion}
 \begin{split}
 D(\sigma\Vert\sigma_*^{\otimes n})
   ={}&n h_2((1-p)q_*)-S(\sigma)\\
     &+(1-p)h_2'((1-p)q_*)\sum_i(q_i-q_*).
 \end{split}
\end{equation}
Substituting the stationarity relation
$(1-p)h_2'((1-p)q_*)=g_p'(q_*)$ and the definitions
of $R_n,b_p$ proves \eqref{eq:exact-deficit}.
For pure $\tau$, the fixed-average-input Holevo
information is zero. Dropping the other nonnegative
terms proves \eqref{eq:radius-bound} for pure inputs;
convexity of relative entropy extends it to all inputs.

Equality $\chi_{\cA_p^{\otimes n}}(\tau)=nC(\cA_p)$
forces $\sigma=\sigma_*^{\otimes n}$. For $p<1$,
amplitude damping is invertible as a linear map: it
rescales the excited population by $1-p$ and the
coherence by $\sqrt{1-p}$, while adding $p\tau_{11}$
to the ground population. Hence $\tau=\tau_*^{\otimes n}$.
Conversely, the product binary ensemble attains the
optimum at this average input.
\end{proof}
The terms in \eqref{eq:exact-deficit} have distinct
origins: displacement of the output from the optimal
product state, excess in the output-entropy convex
roof, and local convexity remainders. The first further
decomposes as
\begin{equation}\label{eq:correlation-decomposition}
 D(\sigma\Vert\sigma_*^{\otimes n})
 =\sum_i S(\sigma_i)-S(\sigma)+\sum_i D(\sigma_i\Vert\sigma_*).
\end{equation}
Here $\sum_i S(\sigma_i)-S(\sigma)$ is the total output
correlation. Equation \eqref{eq:radius-bound} identifies
$\sigma_*^{\otimes n}$ as an information-radius center at every blocklength, in the relative-entropy characterization of Holevo information
\cite{schumacher-optimal}. The uniqueness conclusion
concerns the average input; it allows different optimal
ensembles and collective decoders. The identity is a
finite-block Holevo statement and does not by itself
give a strong converse for nonvanishing error.

\section{Entanglement cost}\label{sec:simulation}
We establish the simulation statements in the parallel LOCC
model of \eqref{eq:simulation-model}. We first derive a
Choi-state criterion under which the general simulation
bounds coincide. Combining this criterion with the
communication results of \cref{sec:capacity} gives the
unified theorem for qubit-to-qubit channels of Kraus rank
at most two and extends all its conclusions to the
pure-output class. We then evaluate the cost for amplitude damping and
construct an explicit two-outcome Kraus representation
that attains the uniform formation bound.

\subsection{Channel-simulation bounds and a Choi-state criterion}
\label{sec:simulation-criterion}
The channel-simulation theorem states
\cite[Theorem~12 and Corollary~17]{op:BBCW}
\begin{equation}\label{eq:berta-regularization}
 \Epar(\cN)=\lim_{n\to\infty}\frac1n\Ech(\cN^{\otimes n}),
 \qquad
 \EC(J_{\cN})\le\Epar(\cN)\le\Ech(\cN).
\end{equation}
For the lower bound, Alice prepares Bell pairs locally
and applies the simulator to one half of each pair.
Diamond-norm convergence implies preparation of
$J_{\cN}^{\otimes n}$ with vanishing trace-distance
error and the same entanglement resource. The upper
bound is supplied by the uniform channel-simulation
theorem: a single protocol must work for every input.

For qubit input and output, concurrence factorization
\cite{op:Konrad} gives
\begin{equation}\label{eq:concurrence-factorization}
 c((\id\otimes\cN)(\psi))=c(J_{\cN})c(\psi)
 \quad\text{for every pure two-qubit input }\psi.
\end{equation}
Here $c$ denotes two-qubit concurrence, for which
$\EF(\rho)=e(c(\rho))$.
Since $e$ is increasing and pure-state concurrence
is at most one, a Bell input maximizes channel
entanglement of formation:
\begin{equation}\label{eq:qubit-channel-eof}
 \Ech(\cN)=\EF(J_{\cN}).
\end{equation}
This identity holds for every qubit-to-qubit channel,
irrespective of Kraus rank, and appears in
\cite[Eqs.~(76)--(77)]{op:BBCW}. Thus
\begin{equation}\label{eq:qubit-sandwich}
 \EC(J_{\cN})\le\Epar(\cN)\le\EF(J_{\cN}).
\end{equation}
It remains to identify a condition under which the two
Choi-state quantities coincide.

For a qubit-to-qubit channel, the pure-output condition
is equivalent to the existence of a nonzero product vector
in $\ker J_{\cN}$. Indeed, for unit vectors $a,b$,
\begin{equation}\label{eq:choi-product}
 \bra{a,b}J_{\cN}\ket{a,b}
       =\frac12\bra b\cN(\proj{\bar a})\ket b,
\end{equation}
where conjugation is taken in the basis defining $\Phi_2$.
For a positive operator, zero expectation is equivalent
to membership in the kernel. A trace-one qubit state
has a zero eigenvalue if and only if it is pure. These
observations prove the equivalence. Allowing a mixed
input gives the same class: every positive-weight pure
component of an input yielding a pure output must yield
that output.

By \cref{cor:state-examples}, the product-kernel condition
implies $\EC(J_{\cN})=\EF(J_{\cN})$, so the bounds in
\eqref{eq:qubit-sandwich} coincide. The same corollary
also makes $J_{\cN}$ strongly additive with every
bipartite partner state. Both properties will be used
in the channel theorem below.

\subsection{Qubit channels of Kraus rank at most two}
\label{sec:two-kraus-simulation}
For Kraus rank at most two, the communication criterion
of \cref{cor:two-kraus} and the Choi-state criterion
above hold simultaneously. We collect their consequences
in the following theorem.

\begin{theorem}[Strong additivity for qubit channels of Kraus rank at most two]
\label[theorem]{thm:two-kraus-unified}
Let $\cN:\cL(\mathbb C^2)\to\cL(\mathbb C^2)$ have
Kraus rank at most two. For every finite-dimensional
partner channel $\cM$,
\begin{align}
 \chi(\cN\otimes\cM)&=\chi(\cN)+\chi(\cM),
 \label{eq:two-kraus-holevo}\\
 \Ech(\cN\otimes\cM)&=\Ech(\cN)+\Ech(\cM),
 \label{eq:two-kraus-formation}\\
 \Epar(\cN\otimes\cM)&=\Epar(\cN)+\Epar(\cM).
 \label{eq:two-kraus-cost-add}
\end{align}
Moreover,
\begin{align}
 C(\cN)&=\chi(\cN),\label{eq:two-kraus-capacity}\\
 \Epar(\cN)&=\EC(J_{\cN})=\EF(J_{\cN})=\Ech(\cN).
 \label{eq:two-kraus-cost}
\end{align}
For every integer $n\ge1$,
\begin{equation}\label{eq:two-kraus-powers}
 \Ech(\cN^{\otimes n})=\EF(J_{\cN}^{\otimes n})
       =n\EF(J_{\cN}).
\end{equation}
\end{theorem}
\begin{proof}
\emph{Classical communication.}
\Cref{cor:two-kraus} gives \eqref{eq:two-kraus-holevo},
\eqref{eq:two-kraus-capacity}, and additivity of the
regularized classical capacity.

\emph{Entanglement cost.}
The Kraus-rank bound implies $\rank J_{\cN}\le2$.
By \cref{cor:state-examples}, its kernel contains a
nonzero product vector and
$\EC(J_{\cN})=\EF(J_{\cN})$. The bounds in
\eqref{eq:qubit-sandwich} therefore coincide.
Together with \eqref{eq:qubit-channel-eof}, this proves
\eqref{eq:two-kraus-cost}.

\emph{Strong formation additivity.}
\Cref{lem:channel-subadditivity} gives the upper bound
in \eqref{eq:two-kraus-formation} for arbitrary
finite-dimensional channels. For the reverse bound,
use a Bell input for $\cN$ and an independent pure
input with reference attaining $\Ech(\cM)$.
The output is $J_{\cN}\otimes\sigma$, where
$\EF(\sigma)=\Ech(\cM)$. Strong state additivity
from \cref{cor:state-examples} and
\eqref{eq:qubit-channel-eof} yield
\[
 \Ech(\cN\otimes\cM)
 \ge\EF(J_{\cN}\otimes\sigma)
 =\EF(J_{\cN})+\EF(\sigma)
 =\Ech(\cN)+\Ech(\cM).
\]
This proves \eqref{eq:two-kraus-formation}.

\emph{Strong cost additivity.}
Repeated application of \eqref{eq:two-kraus-formation}
to the $\cN$ factors gives
\[
 \Ech(\cN^{\otimes n}\otimes\cM^{\otimes n})
 =n\Ech(\cN)+\Ech(\cM^{\otimes n}).
\]
Divide by $n$ and use \eqref{eq:berta-regularization}
and \eqref{eq:two-kraus-cost}. This proves
\eqref{eq:two-kraus-cost-add}, without an additivity
assumption on $\cM$. Iterating formation additivity and
state additivity separately gives \eqref{eq:two-kraus-powers}.
\end{proof}

\begin{corollary}[Qubit channels admitting a pure output]
\label[corollary]{cor:pure-output}
Let $\cN$ be a qubit-to-qubit channel with
$\cN(\proj\psi)=\proj\phi$ for some unit vectors
$\psi,\phi$. Then all conclusions of
\cref{thm:two-kraus-unified} hold. In addition,
\[
 C(\cN\otimes\cM)=C(\cN)+C(\cM)
\]
for every finite-dimensional partner channel $\cM$.
\end{corollary}
\begin{proof}
\label{par:pure-output-extension}
\Cref{lem:pure-output-stinespring} gives
\eqref{eq:two-kraus-holevo}, \eqref{eq:two-kraus-capacity},
and additivity of the regularized classical capacity.
In the formation and simulation parts of the proof of
\cref{thm:two-kraus-unified}, the rank bound is used only
to obtain a product vector in $\ker J_{\cN}$.
Equation \eqref{eq:choi-product} supplies this property
under the pure-output hypothesis, so the same proof applies.
\end{proof}

To see that the pure-output hypothesis is genuinely more general than the Kraus-rank-two assumption, suppose $a,b,c>0$ with $a+b+c=1$,
the Kraus operators
\[
 K_0=\begin{pmatrix}1&0\\0&\sqrt a\end{pmatrix},\qquad
 K_1=\sqrt b\,\ket0\bra1,\qquad
 K_2=\sqrt c\,\ket1\bra1
\]
are linearly independent and define a trace-preserving channel
$\cN$ satisfying $\cN(\proj0)=\proj0$. Thus $\cN$ has
Kraus rank three and falls within \cref{cor:pure-output}.
Its Choi state satisfies $J_{\cN}\ket{01}=0$ and
$\bra{00}J_{\cN}\ket{11}=\sqrt a/2$.
Equation \eqref{eq:qubit-kernel} therefore gives
\[
 \Epar(\cN)=\EF(J_{\cN})=e(\sqrt a)
       =h_2\!\left(\frac{1+\sqrt{1-a}}2\right).
\]

\subsection{Entanglement cost of amplitude damping}
\label{sec:ad-simulation}
Amplitude damping has Kraus rank at most two, so
\cref{thm:two-kraus-unified} reduces its entanglement
cost to the formation of its Choi state. The latter
can be evaluated explicitly.

\begin{theorem}[Entanglement cost of amplitude damping]
\label[theorem]{thm:ad-cost}
For every $p\in[0,1]$,
\begin{equation}\label{eq:main-ad-cost}
\Epar(\cA_p)=\EC(J_{\cA_p})=\EF(J_{\cA_p})=\Ech(\cA_p)=h_2\!\left(\frac{1+\sqrt p}{2}\right).
\end{equation}
For every integer $n\ge1$,
\begin{equation}\label{eq:main-ad-powers}
 \Ech(\cA_p^{\otimes n})=\EF(J_{\cA_p}^{\otimes n})
       =n h_2\!\left(\frac{1+\sqrt p}{2}\right).
\end{equation}
The arbitrary-partner additivity identities of
\cref{thm:two-kraus-unified} apply with $\cN=\cA_p$.
\end{theorem}
\begin{proof}
In the computational basis of reference and output,
\begin{equation}\label{eq:ad-choi}
 J_{\cA_p}=\frac12
 \begin{pmatrix}
 1&0&0&\sqrt{1-p}\\
 0&0&0&0\\
 0&0&p&0\\
 \sqrt{1-p}&0&0&1-p
 \end{pmatrix}.
\end{equation}
Its kernel contains $\ket{01}$. The two-qubit identity
\eqref{eq:qubit-kernel} evaluates its formation as
\begin{equation}\label{eq:ad-choi-formation}
 \EF(J_{\cA_p})=e(\sqrt{1-p})
       =h_2\!\left(\frac{1+\sqrt p}{2}\right).
\end{equation}
Substitution into \cref{thm:two-kraus-unified} gives
\eqref{eq:main-ad-cost}, \eqref{eq:main-ad-powers},
and the arbitrary-partner identities.
\end{proof}
At $p=0,1$, the formula respectively gives the cost of
the identity channel (one ebit per use) and a constant
channel (zero). The cost equalities are asymptotic
vanishing-error statements; \eqref{eq:main-ad-powers}
is an exact finite-copy formation identity. A simulation
from a single Choi copy, or against adaptive tests, is
not asserted.

\paragraph{An explicit two-outcome Kraus representation.}
The general proof uses the uniform representation
argument in \cref{app:kraus}. For amplitude damping,
a fixed two-outcome instrument gives the formation
upper bound directly. Given a finite Kraus
representation $\mathbf L=\{L_s\}$, define
\begin{equation}\label{eq:instrument-cost}
 f_{\mathbf L}(\tau)=\sum_s t_s
 S(L_s\tau L_s^\dagger/t_s),\qquad
 t_s=\Tr(L_s\tau L_s^\dagger).
\end{equation}
Branches of zero weight are omitted. The entropy
perspective makes $f_{\mathbf L}$ continuous and
concave in $\tau$. Applied to a purification of $\tau$,
the Kraus operators give a pure-state decomposition of
the reference--output state with average entanglement
$f_{\mathbf L}(\tau)$. This quantity consequently
upper-bounds the entanglement of formation of that state.

\begin{proposition}[Uniform formation bound for amplitude damping]
\label[proposition]{prop:binary-instrument}
Set
\begin{equation}\label{eq:binary-instrument}
 L_\pm=\frac{K_0\pm K_1}{\sqrt2}
       =\frac1{\sqrt2}\begin{pmatrix}
           1&\pm\sqrt p\\0&\sqrt{1-p}
         \end{pmatrix}.
\end{equation}
Then
\begin{equation}\label{eq:binary-max}
 \max_{\tau\in\cD(\mathbb C^2)}f_{\mathbf L}(\tau)
     =f_{\mathbf L}(I/2)=h_2\!\left(\frac{1+\sqrt p}{2}\right).
\end{equation}
For every integer $n\ge1$ and every block input $\tau$,
\begin{equation}\label{eq:binary-uniform-block}
 f_{\mathbf L^{\otimes n}}(\tau)
       \le n h_2\!\left(\frac{1+\sqrt p}{2}\right).
\end{equation}
Together with Choi-state additivity, this gives an
explicit certificate for \eqref{eq:main-ad-powers}.
\end{proposition}
\begin{proof}
The operators $L_+,L_-$ are a unitary mixing of
$K_0,K_1$, so they represent the same channel. Writing
$z=\tau_{01}$, their branch weights and determinants are
\begin{equation}\label{eq:binary-traces}
 t_\pm=\frac12\pm\sqrt p\,\Re z,\qquad
 \det(L_\pm\tau L_\pm^\dagger)=\frac{1-p}{4}\det\tau.
\end{equation}
For a qubit state $\tau$, the operator $I-\tau$ is
a state with the same determinant. This substitution
exchanges $t_+$ and $t_-$, and hence the normalized
branch spectra. Thus $f_{\mathbf L}(I-\tau)=f_{\mathbf L}(\tau)$.
Concavity implies
\[
 f_{\mathbf L}(\tau)
 =\tfrac12[f_{\mathbf L}(\tau)+f_{\mathbf L}(I-\tau)]
 \le f_{\mathbf L}(I/2)=h_2\!\left(\frac{1+\sqrt p}{2}\right).
\]
Continuity includes zero-weight branches and $p=0,1$.

For an arbitrary input on $n$ channel uses, apply the
product instrument and store its outcomes in classical
registers $S_1,\ldots,S_n$. Subadditivity and removal
of classical conditioning give
\begin{equation}\label{eq:binary-product}
 f_{\mathbf L^{\otimes n}}(\tau)
  =S(B^n|S^n)\le\sum_i S(B_i|S_i)
  =\sum_i f_{\mathbf L}(\tau_i)
  \le nh_2\!\left(\frac{1+\sqrt p}{2}\right).
\end{equation}
Each marginal $B_iS_i$ is the output of the local
instrument on $\tau_i$, because the remaining maps
are trace preserving. The bound therefore holds
uniformly for every purified block input. For
$\tau=(I/2)^{\otimes n}$, the instrument factors across
uses, and \eqref{eq:binary-max} gives equality in
\eqref{eq:binary-uniform-block}. The associated
reference--output state is $J_{\cA_p}^{\otimes n}$.
By \cref{cor:state-examples} and
\eqref{eq:ad-choi-formation}, its entanglement of
formation equals the same value. This establishes
the matching lower bound without using the general
uniform-representation argument.
\end{proof}
The Kraus representation \eqref{eq:binary-instrument}
corresponds to a Hadamard-basis measurement of the
dilation environment. Its symmetry under
$\tau\mapsto I-\tau$ identifies $I/2$ as a maximizer
of the average conditional entropy. The resulting
bound is explicit at every blocklength; operational
achievability still uses the channel-simulation
theorem \eqref{eq:berta-regularization}.

\section{Discussion}\label{sec:discussion}
We have established a common additivity criterion for classical communication and channel simulation.
For every qubit-to-qubit channel admitting a pure output,
the classical capacity equals the single-use Holevo
information, and the parallel entanglement cost equals
the entanglement of formation of the normalized Choi
state. The corresponding Holevo, formation, and cost
quantities are additive with arbitrary finite-dimensional
partners. For amplitude damping, these statements give
the explicit capacity and cost formulas, together with a finite-block
deficit identity and a uniform binary Kraus construction.

The underlying state theorem applies in arbitrary
finite dimensions. Excluding the joint orthogonal
sector $A_1\otimes B_1$ suffices for strong
superadditivity of entanglement of formation. The proof
uses a triangular block entropy inequality to control
a pure extension and a fixed-expectation decomposition
to bound its constrained marginal. For communication,
the theorem applies to an entire output--environment
state family; for simulation, it applies to the
reference--output Choi state. For qubit-to-qubit channels,
these two support conditions are equivalent to admitting
a pure output, as expressed in
\eqref{eq:pure-output-equivalence}. Thus the communication
and simulation conclusions hold for the same class,
including all channels of Kraus rank at most two.

The support hypothesis is sufficient and leaves room
for several extensions. In higher dimensions, a
single product vector in the kernel does not exclude
the full subspace $A_1\otimes B_1$. Generalized
amplitude damping at finite temperature does not
in general satisfy the hypothesis \cite{khatri}.
A quantitative bound in terms of the weight in the
excluded sector could extend the present argument
to such channels. A R\'enyi version of the matrix
inequality would be another direction, potentially
relevant to strong-converse bounds. Neither extension
is established here.

The operational conclusions concern vanishing-error
classical communication and parallel LOCC simulation.
They leave strong-converse thresholds, reliability
functions, and simulation against adaptive tests
\cite{op:Wilde2018} unresolved. The comparisons with
quantum and entanglement-assisted capacities in
\cref{app:operational} use different free resources.
A reversibility statement would require matching
resource conventions, while the tradeoff between classical
communication and entanglement would require a joint
optimization over protocols.

\section*{Acknowledgments}
This work was supported by the National Key R\&D Program of China (Grant No.~2024YFE0102500), the National Natural Science Foundation of China (Grant. No.~92576114, 12447107), the Guangdong Provincial Quantum Science Strategic Initiative (Grant No.~GDZX2403008, GDZX2503001), and the Guangdong Provincial Key Lab of Integrated Communication, Sensing and Computation for Ubiquitous Internet of Things (Grant No.~2023B1212010007), and by the Guangdong Basic and Applied Basic Research Foundation (Grant Nos.~2026A1515030035 and 2025A1515110223). 

\paragraph{Use of AI-assisted tools.}
Frontier Large Language Models were used for the proof of~\Cref{thm:support} and~\Cref{thm:triangle}. The authors rigorously verified the statement and proof, generalized the application of~\Cref{thm:triangle}, and prepared the manuscript with the help of ChatGPT-6 Astra in writing. The authors take full responsibility for the final text and all mathematical claims contained therein.

\paragraph{Lean formalization.}
 The Lean~4.30.0~\cite{deMouraUllrich2021} source code and manuscript-to-Lean correspondence are available in the \href{https://github.com/QuAIR/Classical-Capacity-Qubit-Amplitude-Damping}{accompanying GitHub repository}~\cite{ClassicalCapacityADLean2026}, using Mathlib~4.30.0~\cite{mathlib2020} and Lean-QIT~\cite{ZhuEtAl2026LeanQIT}.

\bibliographystyle{unsrt}
\bibliography{ref}

\appendix
\section{Proof of the triangular block entropy inequality}
\label[appendix]{app:triangle}

We prove \cref{thm:triangle} by expressing the entropy
difference as an integral of log-determinants. Convexity
in two matrix arguments, with the third fixed, permits
dephasing followed by a scalar reduction. Polar
decomposition then identifies this difference with the
triangular block expression. We first treat square
blocks and extend the result to rectangular blocks by
isometric embedding. Throughout this appendix, entropy
is expressed using natural logarithms until the final
conversion to bits.

\Needspace{11\baselineskip}
\subsection{Convexity of a log-determinant functional}

\begin{lemma}
\label[lemma]{lem:logdet}
For an integer $d\ge1$ and a fixed $J\in\mathbb C^{d\times d}$ with
$J^\dagger=-J$, the function
\begin{equation}
 L_J(G)=\ln|\det(G+J)|-\ln\det G
 \label{eq:logdet-function}
\end{equation}
is convex on the positive definite Hermitian matrices, and strictly convex if $J$ is invertible.
\end{lemma}
\begin{proof}
The matrix $G+J$ is invertible: if $(G+J)x=0$, taking
the real part of $x^\dagger(G+J)x$ gives
$x^\dagger Gx=0$, hence $x=0$. Jacobi's formula
therefore applies. For a Hermitian direction $H$,
it yields
\begin{align}
 D^2L_J(G)[H,H]
 =\Tr(G^{-1}HG^{-1}H)
 -\operatorname{Re}\Tr((G+J)^{-1}H(G+J)^{-1}H).
 \label{eq:logdet-hessian}
\end{align}
Set
\[
 E=G^{-1/2}HG^{-1/2},\qquad
 A_0=G^{-1/2}JG^{-1/2},\qquad C_0=(I+A_0)^{-1}.
\]
Here $E$ is Hermitian and $A_0$ is skew-Hermitian.
The identity
\[
 (I+A_0)^\dagger(I+A_0)=I+A_0^\dagger A_0\succeq I,
\]
implies $\normop{C_0}\le1$.
Using $(G+J)^{-1}=G^{-1/2}C_0G^{-1/2}$ and
cyclicity of trace, we bound the Hessian by
\begin{align*}
 \Tr E^2-\operatorname{Re}\Tr(C_0EC_0E)
 &\ge\normF E^2-\big|\Tr[(C_0E)^2]\big|\\
 &\ge\normF E^2-\normF{C_0E}^2\ge0.
\end{align*}
More explicitly, this gives the lower bound
\[
 D^2L_J(G)[H,H]\ge
 (1-\|C_0\|_{\mathrm{op}}^2)\|E\|_F^2.
\]
In the preceding three-step chain, the second inequality is the
Hilbert--Schmidt Cauchy--Schwarz bound $|\Tr X^2|\le\normF X^2$,
applied to $X=C_0E$. Thus, the Hessian is nonnegative
on the positive definite cone, proving convexity.
If $J$ is invertible, then $A_0$ is invertible and
$\normop{C_0}<1$. The last bound is consequently
strict for every nonzero $H$, proving strict convexity.
\end{proof}

\subsection{Convexity of the entropy difference}

For a positive semidefinite matrix $X$ of arbitrary
trace, let $\Hraw(X)=-\Tr X\ln X$, with $0\ln0=0$.
For states, $\Hraw(X)=(\ln2)S(X)$. Given positive
semidefinite $d\times d$ matrices $P,Q,R$, define
\begin{align}
 Z(P,Q,R)&=\begin{pmatrix}P+Q&\sqrt Q\sqrt R\\\sqrt R\sqrt Q&R\end{pmatrix},
 \label{eq:Z-definition}\\
 \cF(P,Q,R)&=\Hraw(Z(P,Q,R))-\Hraw(P)-\Hraw(Q)-\Hraw(R).
 \label{eq:F-definition}
\end{align}
Positivity of $Z$ follows from the factorization
\[
 Z=L^\dagger L,\qquad
 L=\begin{pmatrix}\sqrt P&0\\\sqrt Q&\sqrt R\end{pmatrix}.
\]
In particular, $\Tr Z=\Tr P+\Tr Q+\Tr R$.

For positive semidefinite $X,Y$ of the same dimension
and trace, the entropy difference has the representation
\begin{equation}
 \Hraw(X)-\Hraw(Y)=\int_0^\infty
 [\ln\det(X+sI)-\ln\det(Y+sI)]\,ds.
 \label{eq:entropy-integral}
\end{equation}
Diagonalizing both matrices reduces this identity to
\[
 \int_0^L\ln(s+x)\,ds
 =(L+x)\ln(L+x)-L-x\ln x.
\]
Equal dimensions and traces cancel the divergent
upper-endpoint terms as $L\to\infty$. The integrand
is $O(s^{-2})$ at infinity, while zero eigenvalues
produce only integrable logarithmic singularities
at the lower endpoint.

Apply \eqref{eq:entropy-integral} to $Z\oplus0_d$
and $P\oplus Q\oplus R$, which have the same trace
and dimension $3d$. This gives
\begin{equation}
 \cF(P,Q,R)=\int_0^\infty
 \ln\frac{\det(Z+sI)\,s^d}
 {\det(P+sI)\det(Q+sI)\det(R+sI)}\,ds.
 \label{eq:F-integral}
\end{equation}
For $s>0$, define
\begin{equation}
 G_s=\begin{pmatrix}P+sI&0\\0&R+sI\end{pmatrix},\qquad
 J_s=\begin{pmatrix}0&\sqrt s\sqrt Q\\-\sqrt s\sqrt Q&0\end{pmatrix}.
 \label{eq:GsJs}
\end{equation}
The Schur complement of the lower-right block yields
\begin{align}
 \det(Z+sI)
 &=\det(R+sI)
 \det\!\left(P+Q+sI-\sqrt Q\,R(R+sI)^{-1}\sqrt Q\right)\notag\\
 &=\det(R+sI)
 \det\!\left(P+sI+s\sqrt Q(R+sI)^{-1}\sqrt Q\right)\notag\\
 &=\det(G_s+J_s).
 \label{eq:schur-determinant}
\end{align}
Only $R$ and $(R+sI)^{-1}$ have been commuted.
The final equality follows from the Schur complement
of $G_s+J_s$; the opposite signs in its off-diagonal
blocks give the positive correction. In particular,
both determinants are positive real. The integrand
in \eqref{eq:F-integral} can therefore be written as
\begin{equation}
 L_{J_s}(G_s)-\ln\det(I+Q/s).
 \label{eq:integrand-convex}
\end{equation}
With $Q$ fixed, $J_s$ is fixed and skew-Hermitian,
whereas $G_s$ depends affinely on $(P,R)$.
\Cref{lem:logdet} makes the integrand convex in this
pair. Integrating over $[\varepsilon,L]$ preserves
convexity, and convergence of the integral permits
$\varepsilon\to 0$ and $L\to\infty$.
Hence
\begin{equation}
 (P,R)\longmapsto\cF(P,Q,R)\quad\text{is convex for fixed }Q.
 \label{eq:fixed-middle-convexity}
\end{equation}
The shifts by $sI$ keep $G_s$ positive definite,
including for singular $P,Q,R$. If $Q>0$, each
$J_s$ is invertible, so the integrand is strictly
convex in $(P,R)$. Integrating its strict Jensen
inequality shows that $\cF$ is strictly convex
in this pair whenever $Q>0$.

\subsection{Dephasing and scalar reduction}

For nonnegative scalars, write
$f(x,y,z)=\cF([x],[y],[z])$, where $[x]$ denotes
the one-dimensional matrix with entry $x$.
We claim that
\begin{equation}
 \cF(P,Q,R)\geq f(\Tr P,\Tr Q,\Tr R).
 \label{eq:trace-compression}
\end{equation}
First assume $Q>0$, and choose an eigenbasis in
which $Q=\diag(q_1,\ldots,q_d)$. Dephasing
$\Delta$ in this basis is the average of
conjugations by
\[
 U_\epsilon=\diag(\epsilon_1,\ldots,\epsilon_d),
 \qquad\epsilon_j\in\{-1,1\}.
\]
These unitaries fix $Q$. Since $\cF$ is invariant
under simultaneous unitary conjugation, convexity
in $(P,R)$ gives
\begin{equation}
 \cF(P,Q,R)\geq\cF(\Delta(P),Q,\Delta(R))
 =\sum_j f(p_j,q_j,r_j),
 \label{eq:F-pinching}
\end{equation}
where $p_j,r_j$ are the diagonal entries of $P,R$.
The equality follows by grouping $Z$ into scalar
$2\times2$ blocks.

Cancellation of the trace terms makes $f$ homogeneous:
$f(cx,cy,cz)=cf(x,y,z)$ for $c\ge0$. Moreover,
$\phi(x,z)=f(x,1,z)$ is convex by
\eqref{eq:fixed-middle-convexity}. Set
$p_\Sigma=\sum_jp_j$, $q_\Sigma=\sum_jq_j$, and
$r_\Sigma=\sum_jr_j$. Homogeneity and Jensen's
inequality with weights $q_j/q_\Sigma$ yield
\begin{align}
 \sum_{j}f(p_j,q_j,r_j)
 &=q_{\Sigma}\sum_j\frac{q_j}{q_{\Sigma}}
   \phi\!\left(\frac{p_j}{q_j},\frac{r_j}{q_j}\right)\notag\\
 &\geq q_{\Sigma}\phi\!\left(\frac{p_{\Sigma}}{q_{\Sigma}},\frac{r_{\Sigma}}{q_{\Sigma}}\right)
 =f(p_{\Sigma},q_{\Sigma},r_{\Sigma}).
 \label{eq:scalar-perspective}
\end{align}
Combining \eqref{eq:F-pinching} with
\eqref{eq:scalar-perspective} proves
\eqref{eq:trace-compression} for $Q>0$.

For general positive semidefinite arguments, add
$\delta I$ to each of $P,Q,R$ and let
$\delta\to 0$. Continuity of the matrix
square root and of $\Hraw$ in fixed dimension
extends the inequality to the boundary, without
division by a zero eigenvalue of $Q$.

\subsection{Polar decomposition and normalization}

Suppose first that all blocks are square and have
the same size. Choose polar decompositions
\[
 B=U_B\sqrt{B^\dagger B},\qquad
 D=U_D\sqrt{D^\dagger D},\qquad W=U_B^\dagger U_D.
\]
Extend the polar factors to unitaries on any
null spaces, and put
\[
 P=A^\dagger A,\qquad Q=B^\dagger B,\qquad
 R=WD^\dagger DW^\dagger.
\]
Then
\begin{equation}
 (I\oplus W)T^\dagger T(I\oplus W^\dagger)=Z(P,Q,R).
 \label{eq:polar-gauge}
\end{equation}
To check the off-diagonal block, note that
\[
 B^\dagger DW^\dagger
 =\sqrt Q\,W\sqrt{D^\dagger D}\,W^\dagger
 =\sqrt Q\sqrt R.
\]
The diagonal blocks follow directly, and the
other off-diagonal block is the adjoint.
The polar factors of $B$ and $D$ are independent.

The matrices $P,Q,R$ have traces
$\alpha,\beta,\gamma$. Their normalized entropies
coincide with those of $AA^\dagger/\alpha$,
$BB^\dagger/\beta$, and $DD^\dagger/\gamma$,
respectively, whenever the weights are nonzero.
For example, for $\alpha>0$,
\[
 \Hraw(P)=-\alpha\ln\alpha+\alpha(\ln2)S(P/\alpha).
\]
The scalar matrix $Z(\alpha,\beta,\gamma)$ has
trace one and determinant $\alpha\gamma$, so
its entropy in bits is
$e(2\sqrt{\alpha\gamma})$.
Substituting into \eqref{eq:trace-compression}
and cancelling the scalar weight entropies gives
\begin{align}
 &\cF(P,Q,R)-f(\alpha,\beta,\gamma)\notag\\
 &\quad=(\ln2)\bigl[S(TT^\dagger)-e(2\sqrt{\alpha\gamma})
 -\alpha S(P/\alpha)-\beta S(Q/\beta)-\gamma S(R/\gamma)\bigr]
 \geq0.
 \label{eq:normalization-cancellation}
\end{align}
This proves \eqref{eq:triangle-bound} for square
blocks with positive weights.

For rectangular blocks, choose a common dimension
at least $\max\{r_1,r_2,c_1,c_2\}$. Let $E_1,E_2$
be isometric embeddings of the row spaces and
$F_1,F_2$ embeddings of the column spaces into
that dimension. Then
\[
 \widetilde T=(E_1\oplus E_2)T(F_1\oplus F_2)^\dagger
\]
has square blocks and the same nonzero singular
values as $T$. Each embedded block also preserves
its nonzero singular values and Frobenius norm.
The square-block inequality therefore proves
the rectangular case.

If a weight vanishes, its block is zero. In fixed
dimension, $0\le S(\rho)\le\log_2d$ ensures
that the weighted entropy tends to zero with
the weight. Continuity includes these cases
and completes the proof of \cref{thm:triangle}.

\subsection{Equality conditions}
\begin{proposition}[Equality for an invertible middle block]
\label[proposition]{prop:triangle-equality}
Suppose $A,B,D$ are square blocks of the same size, $B$ is invertible,
and $\alpha,\beta,\gamma>0$. Equality in \eqref{eq:triangle-bound}
holds if and only if
\begin{equation}\label{eq:triangle-equality}
 A^\dagger A=\frac\alpha\beta B^\dagger B,
 \qquad DD^\dagger=\frac\gamma\beta BB^\dagger.
\end{equation}
Equivalently, there are block-diagonal row and column unitaries
$U,V$ and a matrix $R_0$ with $\normF{R_0}=1$ such that
$UTV=T_{\mathrm{sc}}\otimes R_0$.
\end{proposition}
\begin{proof}
Since $Q=B^\dagger B>0$, the dephasing inequality
\eqref{eq:F-pinching} uses strict convexity.
Equality holds only if all pairs
$(U_\epsilon P U_\epsilon^\dagger,
U_\epsilon R U_\epsilon^\dagger)$ coincide.
Thus $P$ and $R$ must be diagonal in the chosen
eigenbasis of $Q$.

Strict convexity of $\phi(x,z)=f(x,1,z)$ makes
equality in \eqref{eq:scalar-perspective}
equivalent to $p_j/q_j=\alpha/\beta$ and
$r_j/q_j=\gamma/\beta$ for every $j$.
Hence $P=(\alpha/\beta)Q$ and
$R=(\gamma/\beta)Q$. Recalling
$R=WD^\dagger DW^\dagger$, these are precisely
the relations in \eqref{eq:triangle-equality}.

Conversely, these proportionality relations saturate
both inequalities. To see the tensor-product form directly,
write $B=U_B\sqrt Q$. The first relation gives
$A=\sqrt{\alpha/\beta}\,U_A\sqrt Q$; the second gives
$D=\sqrt{\gamma/\beta}\,U_B\sqrt Q\,V_D$ for suitable
unitaries. Multiplying the block rows by $U_A^\dagger,U_B^\dagger$
and the second block column by $V_D^\dagger$ yields
$T_{\mathrm{sc}}\otimes\sqrt{Q/\beta}$, as claimed.
\end{proof}
For singular $B$, the tensor-product form remains
sufficient for equality. The proposition asserts
necessity only for invertible $B$.

\section{Scalar convexity and two-qubit concurrence}
\label[appendix]{app:concurrence}

\subsection{Convexity of the binary-spectrum entropy}
\label{app:scalar}

Let $0<c<1$ and set $r=\sqrt{1-c^2}$.
Differentiating \eqref{eq:E-function} gives
\[
 e'(c)=\frac{c}{2r\ln2}\ln\frac{1+r}{1-r}>0,
 \qquad
 e''(c)=\frac{\ln((1+r)/(1-r))-2r}{2r^3\ln2}\ge0.
\]
The sign of the second derivative follows from
\[
 \frac12\ln\frac{1+r}{1-r}
 =\int_0^r\frac{du}{1-u^2}\ge r.
\]
By continuity, $e$ is increasing and convex
on $[0,1]$. Since $g_p(q)=e(2\sqrt{p(1-p)}q)$,
it is convex in $q$ for every $p\in[0,1]$,
and vanishes identically at $p=0,1$.

\subsection{Concurrence for a product vector in the kernel}\label{app:qubit}

A two-qubit state satisfying $\rho\ket{01}=0$
has the following form in the computational basis:
\begin{equation}
 \rho=\begin{pmatrix}
 a&0&x&z\\0&0&0&0\\x^*&0&b&y\\z^*&0&y^*&d
 \end{pmatrix},\qquad a+b+d=1.
 \label{eq:kernel-state-matrix}
\end{equation}
Positivity implies $|z|\le\sqrt{ad}$; the
entries $x,y$ can be nonzero. Define the spin flip by
\[
 \widetilde\rho=(\sigma_y\otimes\sigma_y)\rho^*
                  (\sigma_y\otimes\sigma_y),\qquad
 \sigma_y=\begin{pmatrix}0&-i\\i&0\end{pmatrix}.
\]
Wootters' formula states \cite{wootters}
\begin{equation}
 \EF(\rho)=e(c(\rho)),\qquad
 c(\rho)=\max\{0,s_1-s_2-s_3-s_4\},
 \label{eq:wootters}
\end{equation}
where $s_1\ge s_2\ge s_3\ge s_4\ge0$ are the
square roots of the eigenvalues of
$\rho\widetilde\rho$.

The $01$ row and the $10$ column of
$\rho\widetilde\rho$ vanish. Expanding its
characteristic polynomial along them leaves two
zero eigenvalues and the two eigenvalues of
\[
 \begin{pmatrix}
 ad+|z|^2&2az\\2dz^*&ad+|z|^2
 \end{pmatrix}.
\]
These are $(\sqrt{ad}+|z|)^2$ and
$(\sqrt{ad}-|z|)^2$.
Using $|z|\le\sqrt{ad}$ gives
\[
 (s_1,s_2,s_3,s_4)
 =(\sqrt{ad}+|z|,\sqrt{ad}-|z|,0,0),
 \qquad c(\rho)=2|z|.
\]
Equation \eqref{eq:wootters} now gives
\eqref{eq:qubit-kernel}. The calculation includes
singular states and arbitrary allowed $x,y$.
\section{Uniform Kraus representations}\label[appendix]{app:kraus}
We establish the uniform formation bound used
to prove subadditivity of channel entanglement
of formation. The argument follows the minimax
construction of \cite[Lemma~14]{op:BBCW}, using
compactness to reduce to finitely many Kraus
representations. For amplitude damping,
\cref{prop:binary-instrument} gives an explicit
representation attaining the bound.

Let $\mathbf K=\{K_i\}$ be a finite Kraus
representation of $\cN$, with $f_{\mathbf K}$
defined by \eqref{eq:instrument-cost}. For a
purification $\psi^\tau$ of $\tau$, Kraus freedom
implies
\begin{equation}\label{eq:kraus-inf}
 \inf_{\mathbf K}f_{\mathbf K}(\tau)
          =\EF((\id\otimes\cN)(\psi^\tau)).
\end{equation}
Each Kraus representation gives a pure-state
ensemble of the reference--output state. Conversely,
the dilation environment purifies this state,
so every pure-state ensemble is induced by a
rank-one measurement on the environment. The
resulting linear combinations of Kraus operators
represent $\cN$; unused environmental directions
can be completed without changing the ensemble
for the chosen input.

\begin{lemma}[Uniform Kraus representation]\label[lemma]{lem:uniform-kraus}
For every finite-dimensional channel,
\begin{equation}\label{eq:kraus-minimax}
 \Ech(\cN)=\inf_{\mathbf K}\max_{\tau\in\cD(X)}f_{\mathbf K}(\tau).
\end{equation}
Equivalently, for every $\delta>0$, a single finite representation
satisfies $f_{\mathbf K}(\tau)\le\Ech(\cN)+\delta$ for all inputs.
\end{lemma}
\begin{proof}
Write $E=\max_\tau\inf_{\mathbf K}f_{\mathbf K}(\tau)$.
By \eqref{eq:kraus-inf}, $E=\Ech(\cN)$.
The minimax inequality gives
$E\le\inf_{\mathbf K}\max_\tau f_{\mathbf K}(\tau)$.

For the converse, fix $\delta>0$. At each
$\tau$, choose a finite representation with
$f_{\mathbf K}(\tau)<E+\delta$.
Continuity preserves this bound in a neighborhood
of $\tau$, and compactness of $\cD(X)$ gives
a finite subcover. Its representations
$\mathbf K^{(1)},\ldots,\mathbf K^{(m)}$ satisfy
\[
 \min_{1\le j\le m} f_{\mathbf K^{(j)}}(\tau)<E+\delta
 \quad\text{for every }\tau.
\]
The objective below is continuous and concave
in $\tau$ and affine in the probability vector
$\boldsymbol w$. The finite-dimensional minimax
theorem therefore gives
\begin{equation}\label{eq:finite-minimax}
 \min_{\boldsymbol w\in\Delta_m}\max_\tau
       \sum_jw_jf_{\mathbf K^{(j)}}(\tau)
 =\max_\tau\min_j f_{\mathbf K^{(j)}}(\tau)
 \le E+\delta.
\end{equation}
For a minimizing $\boldsymbol w$, the operators
$\{\sqrt{w_j}K_i^{(j)}\}_{j,i}$ form a finite
Kraus representation of $\cN$. Its conditional
entropy is the weighted sum in
\eqref{eq:finite-minimax}. Taking
$\delta\to 0$ proves \eqref{eq:kraus-minimax}.
\end{proof}

\begin{lemma}[Subadditivity of channel entanglement of formation]
\label[lemma]{lem:channel-subadditivity}
For any finite-dimensional channels $\cN,\cM$,
\begin{equation}\label{eq:channel-subadditivity}
 \Ech(\cN\otimes\cM)\le\Ech(\cN)+\Ech(\cM).
\end{equation}
\end{lemma}
\begin{proof}
Choose representations $\mathbf K,\mathbf L$
whose uniform bounds are within $\delta$ of
$\Ech(\cN),\Ech(\cM)$, respectively. For a
joint input $\tau_{XY}$, apply the product
instrument and store its outcomes in classical
registers $I,J$. Subadditivity and strong
subadditivity of entropy give
\begin{align}
 f_{\mathbf K\otimes\mathbf L}(\tau_{XY})
   &=S(BB'|IJ)\nonumber\\
   &\le S(B|IJ)+S(B'|IJ)\nonumber\\
   &\le S(B|I)+S(B'|J)
     =f_{\mathbf K}(\tau_X)+f_{\mathbf L}(\tau_Y).
       \label{eq:product-kraus-bound}
\end{align}
Trace preservation of the other channel ensures
that $BI$ and $B'J$ are the corresponding local
instrument outputs. Applied to a purification
of $\tau_{XY}$, the product Kraus operators
give a pure-state decomposition of the joint
reference--output state. Its entanglement of
formation is thus at most
$f_{\mathbf K\otimes\mathbf L}(\tau_{XY})$.
Maximizing over $\tau$ and letting
$\delta\to 0$ proves the claim.
\end{proof}
The matching lower bound for the channels considered
in this work is supplied by strong Choi-state additivity
in \cref{thm:two-kraus-unified,cor:pure-output}, covering
all qubit-to-qubit channels admitting a pure output.

\section{Established operational rates and asymptotic comparisons}\label[appendix]{app:operational}
This appendix records the rates plotted in
\cref{fig:ad-rates} and their high-damping limits.
The quantum and entanglement-assisted capacity
formulas, as well as the reverse Shannon theorem,
are established results. The additivity theorems
in the main text identify the single-use Holevo
information and the Choi-state entanglement of
formation with the corresponding unassisted
capacity and state and parallel channel costs.


\subsection{Quantum and entanglement-assisted capacities}
The unassisted quantum capacity $Q(\cN)$ is
the supremum of entanglement-transmission rates
achievable with vanishing error. An encoder,
$n$ channel uses, and a decoder must preserve
a maximally entangled state with an inaccessible
reference. The rate is measured in qubits per
use, without preshared entanglement or auxiliary
classical communication. The entanglement-assisted
classical capacity $C_E(\cN)$ measures classical
bits per use when shared entanglement is free.

The resource convention also matters for
simulation. The cost $\Epar$ counts ebits with
classical communication free. If forward
classical communication is counted and
unrestricted shared entanglement is free,
write the parallel simulation cost as
$C_{\mathrm{sim}}^{\mathrm{cl,EA}}(\cN)$.
The universal quantum reverse Shannon theorem
identifies it with
\begin{equation}\label{eq:qrst}
 C_{\mathrm{sim}}^{\mathrm{cl,EA}}(\cN)=C_E(\cN)
       =\max_{\rho_X}I(R;B)_{(\id\otimes\cN)(\psi^\rho)},
\end{equation}
where $\psi^\rho$ purifies $\rho$
\cite{op:BSST,op:BCR,op:BDHSW}. The free shared
resource here includes the entanglement spread
needed for uniform simulation and is not restricted
to a fixed tensor power of ebits. These two
simulation costs optimize different resources;
a protocol minimizing one need not minimize the other.

For the diagonal input $\tau_q$ in \eqref{eq:tauq}, denote
its coherent information and reference--output mutual
information by
\begin{equation}\label{eq:IQ-G}
 \begin{split}
 I_p(q)&=h_2((1-p)q)-h_2(pq),\\
 G_p(q)&=h_2(q)+h_2((1-p)q)-h_2(pq).
 \end{split}
\end{equation}
The known capacity formulas are
\begin{align}
 Q(\cA_p)&=
 \begin{cases}
 \displaystyle\max_{0\le q\le1}I_p(q),&0\le p\le\tfrac12,\\
 0,&\tfrac12\le p\le1,
 \end{cases}\label{eq:AD-Q}\\
 C_E(\cA_p)=C_{\mathrm{sim}}^{\mathrm{cl,EA}}(\cA_p)
   &=\max_{0\le q\le1}G_p(q).\label{eq:AD-CE}
\end{align}
The quantum-capacity formula follows from
degradability and antidegradability
\cite{giovannetti-fazio,khatri,op:DS}; the
entanglement-assisted formula was evaluated in
\cite[Sec.~III.B]{op:BSST}.

To see these reductions, the complement in
\eqref{eq:dilation} is $\cA_p^c=\cA_{1-p}$,
and amplitude-damping channels compose as
$\cA_s\circ\cA_p=\cA_{p+s-ps}$.
For $p\le1/2$, the channel
$\cA_{(1-2p)/(1-p)}$ degrades the output to
the environment; for $p\ge1/2$,
$\cA_{(2p-1)/p}$ antidegrades the environment
to the output. Degradability gives the
single-letter coherent-information formula,
whereas antidegradability forces $Q=0$.

For degradable channels, coherent information
is concave in the input. Together with phase
covariance, this permits restriction to diagonal
inputs. Channel mutual information is concave
for every channel, giving the same reduction
for $C_E$. At the diagonal input $\tau_q$,
the purified dilation satisfies
\[
 S(R)=h_2(q),\qquad S(B)=h_2((1-p)q),\qquad S(E)=h_2(pq).
\]
Hence $I_p(q)=S(B)-S(E)$ and $G_p(q)=I(R;B)$,
which give \eqref{eq:AD-Q} and \eqref{eq:AD-CE}.
The simulation equality follows from
\eqref{eq:qrst}.

\subsection{Relations at a fixed input excitation}
We compare the classical capacity in \eqref{eq:main-capacity}
and the entanglement cost in \eqref{eq:main-ad-cost}
with the established rates above. The single-use Holevo
evaluation, the Choi-state formation formula, and the
channel-simulation upper bound come from
\cite{giovannetti-fazio,wootters,op:BBCW}. The additivity
results identify those expressions with the corresponding
asymptotic capacity and cost.

At zero damping,
\[
 \bigl(C(\cA_0),Q(\cA_0),\Epar(\cA_0),C_E(\cA_0)\bigr)
       =(1,1,1,2);
\]
at $p=1$, all four rates vanish. At $p=1/2$, scalar
optimization gives
\begin{equation}\label{eq:midpoint-values}
 \bigl(C(\cA_{1/2}),Q(\cA_{1/2}),
       \Epar(\cA_{1/2}),C_E(\cA_{1/2})\bigr)
       \simeq(0.47173,0,0.60088,1),
\end{equation}
where $C_E=1$ is exact because
$G_{1/2}(q)=h_2(q)$. The maximizing excitation
need not be the same for different rates.

At a common excitation $q$, the opposite-phase
binary ensemble has environment Holevo information
$F_{1-p}(q)$, since $g_{1-p}=g_p$. Thus
\begin{align}
 F_p(q)-F_{1-p}(q)&=I_p(q),\label{eq:fixed-input-diff}\\
 G_p(q)-F_p(q)&=h_2(q)-F_{1-p}(q)\ge0,\label{eq:EA-diff}\\
 G_p(q)+G_{1-p}(q)&=2h_2(q).\label{eq:EA-complement}
\end{align}
Data processing bounds the environment Holevo
information by the input ensemble's Holevo
information, giving \eqref{eq:EA-diff}.
Together with $F_{1-p}(q)\ge0$, the identities
imply $Q(\cA_p)\le C(\cA_p)\le C_E(\cA_p)$.
The equalities themselves hold at a common
input excitation; they do not generally survive
separate optimization of their terms.

\subsection{Asymptotic expansion in the high-damping limit}
Set $\varepsilon=1-p\to 0$ and
$L=\log_2(1/\varepsilon)$. Then
\begin{equation}\label{eq:high-damping}
 \begin{split}
 C(\cA_{1-\varepsilon})&=\tfrac14\varepsilon L+O(\varepsilon),\\
 \Epar(\cA_{1-\varepsilon})&=\tfrac14\varepsilon L+O(\varepsilon),\\
 C_E(\cA_{1-\varepsilon})&=\varepsilon L(1+o(1)).
 \end{split}
\end{equation}
Indeed, the smaller eigenvalue defining
$g_{1-\varepsilon}(q)$ is
$\varepsilon q^2+O(\varepsilon^2)$ uniformly
in $q\in[0,1]$. Using
$h_2(x)=x\log_2(1/x)+O(x)$ gives
\[
 F_{1-\varepsilon}(q)=\varepsilon q(1-q)L+O(\varepsilon)
\]
uniformly in $q$, since $q\log(1/q)$ and
$q^2\log(1/q)$ are bounded. Maximizing
$q(1-q)$ gives the first line of
\eqref{eq:high-damping}. The cost expansion
follows from
$(1-\sqrt{1-\varepsilon})/2
=\varepsilon/4+O(\varepsilon^2)$.

For the entanglement-assisted capacity, first
fix $q\in(0,1)$. Then
$G_{1-\varepsilon}(q)/(\varepsilon L)\to q$.
Taking $q\to 1$ after this limit gives
the lower bound on the optimized leading
coefficient. For the upper bound, uniformly
in $q$,
\[
 h_2(q)-h_2((1-\varepsilon)q)
 \le\varepsilon q\log_2\frac1{(1-\varepsilon)q}=O(\varepsilon),
\]
and $h_2(\varepsilon q)\le h_2(\varepsilon)$
for sufficiently small $\varepsilon$.
This proves the last line of
\eqref{eq:high-damping} and yields
\begin{equation}\label{eq:high-damping-ratios}
 \lim_{p\to 1}\frac{C_E(\cA_p)}{C(\cA_p)}=4,
 \qquad
 \lim_{p\to 1}\frac{\Epar(\cA_p)}{C(\cA_p)}=1.
\end{equation}
The factor four for $C_E/\chi$ was obtained in
\cite[Sec.~III.B and Fig.~7]{op:BSST}; the
capacity theorem identifies its denominator
with $C$. The second limit also uses the
entanglement-cost equality. Both are limits
of numerical rate ratios in different resource
models; the ratios are undefined at $p=1$.

For $1/2\le p<1$, the quantum capacity vanishes
while the parallel entanglement cost remains
positive. Since the transmission task is
unassisted and simulation allows LOCC, this
comparison alone does not imply irreversibility
under a common class of free operations.
A two-way-assisted distillation rate would
provide the corresponding transmission
benchmark. That rate, adaptive simulation,
and the joint classical-bit--ebit tradeoff
are not determined here.

\end{document}